\documentclass[11pt,a4paper]{article}

\usepackage{placeins}
\usepackage{amsmath,amssymb,amsfonts}

\usepackage{mathrsfs}

\usepackage{hyperref}

\usepackage{geometry}

\usepackage{authblk}

\usepackage{textcomp}

\usepackage{color}

\usepackage{amsthm}

\usepackage{bm}

\usepackage{tikz-cd}

\usepackage{mathtools}

\usepackage{tabularx}

\usepackage{booktabs}

\usepackage[
  backend=biber,
  style=numeric,
  sorting=nyt,
    sortcites=true,
    giveninits=true,
  maxbibnames=99
]{biblatex}

\newcommand{\R}{\mathbb{R}}
\newcommand{\dd}{\mathrm{d}}
\newcommand{\Lie}{\mathcal{L}}
\theoremstyle{plain} \newtheorem{theorem}{Theorem}[section] \newtheorem{proposition}[theorem]{Proposition} \newtheorem{corollary}[theorem]{Corollary}  \theoremstyle{definition} \newtheorem{definition}[theorem]{Definition} \theoremstyle{remark} \newtheorem{remark}[theorem]{Remark} \newtheorem{example}[theorem]{Example} 
\newcommand{\Cinfty}{\mathscr{C}^\infty} \newcommand{\T}{\mathrm{T}} \newcommand{\cT}{\mathrm{T}^\ast}  \newcommand*{\inn}[1]{\iota_{#1}} 
     \DeclareMathOperator{\rk}{rank} \DeclareMathOperator{\cork}{corank}   \DeclareMathOperator{\pr}{pr}   
     \newcommand{\parder}[2]{\frac{\partial #1}{\partial #2}}      
\newcommand{\subjclass}[2][]{%
\noindent\textbf{MSC 2020:} #2\par
}

\newcommand{\keywords}[1]{%
\noindent\textbf{Keywords:} #1\par
}

\newcolumntype{L}[1]{>{\raggedright\arraybackslash}p{#1}}

\renewbibmacro*{journal}{%
  \iffieldundef{shortjournal}
    {\iffieldundef{journaltitle}
       {}
       {\printtext[journaltitle]{%
          \printfield[titlecase]{journaltitle}%
          \setunit{\subtitlepunct}%
          \printfield[titlecase]{journalsubtitle}}}}
    {\printfield[journaltitle]{shortjournal}}}
\newcolumntype{Y}{>{\raggedright\arraybackslash}X}
\title{\bf A Guide to Applications of $k$-Contact Geometry\\
in Dissipative Field Equations}

\author{J.~de~Lucas}
\author{J.~Lange}
\author{M.~Krych}

\affil{\small Department of Mathematical Methods in Physics, Faculty of Physics,
University of Warsaw, ul.~Pasteura~5, 02--093 Warsaw, Poland}

\date{}

\begin{document}

\maketitle


\begin{abstract}

We study the practical scope of the $k$-contact Hamilton--De Donder--Weyl formalism as a geometric framework for dissipative field equations. In particular, our work focuses on canonical $k$-contact manifolds on $\bigoplus^k \T^*Q\times\mathbb R^k$ and  $k$-contactifications of exact $k$-symplectic phase spaces. A special two-contactification of exact two-symplectic structures on cotangent bundles is defined and analysed.  We also develop several tools for applications, including splitting results for the Hamilton--De Donder--Weyl equations on $k$-contactifications, regularity conditions for such spaces, criteria for the ultrahyperbolicity, hyperbolicity, or ellipticity of PDEs associated with Hamiltonian $k$-contact systems,  dissipation laws associated with infinitesimal dynamical symmetries, relevance and applications of quadratic dissipative terms in the Hamiltonian, etc. Our methods yield explicit Hamiltonian descriptions for several nonlinear nonconservative PDEs with polynomial  dissipative terms, including damped Klein--Gordon, Allen--Cahn, generalized Burgers,   porous medium equations with linear absorption, complex Ginzburg--Landau, damped nonlinear Schr\"odinger, Fisher--KPP, damped $\phi^4$, damped sine--Gordon, and FitzHugh--Nagumo equations, and many others. Our work also stresses the many further practical applications of this framework. 
\end{abstract}






\section{Introduction}

The geometric formulation of classical field theories has traditionally been developed within symplectic, \(k\)-symplectic, \(k\)-cosymplectic, and multisymplectic geometry \cite{KT_79,LR_89,Arn_89}. In this setting, the Hamilton--De Donder--Weyl (HdDW) formalism provides a covariant finite-dimensional description of field equations with several independent variables and has proved particularly successful for conservative systems \cite{Gun_87,Kij_73,KT_79,Mer_97}. However, many physically relevant partial differential equations fall outside the standard conservative framework. Damping, viscosity, relaxation, gain and loss, reaction terms, nonlinear diffusion, and entropy production are ubiquitous in models arising in mechanics, nonlinear optics, pattern formation, transport, and mathematical biology. Typical examples include damped wave and Klein--Gordon equations, reaction--diffusion systems, the porous medium equation with linear absorption, complex amplitude equations, and excitable-media models \cite{AllenCahn1979,Vazquez2007,AransonKramer2002,Murray2002,Nagumo1962}.

Contact geometry provides a natural geometric language for nonconservative mechanics by allowing the Hamiltonian to vary along the dynamics \cite{Bra_17,BCT_17,LS_17}. In the field theory setting, contact geometry can be extended to \(k\)-contact geometry, which generalises contact Hamiltonian methods to field theories with several independent variables. This program was initiated for dissipative Hamiltonian field theories in \cite{GGMRR_20}, further developed in the Lagrangian and Skinner--Rusk settings in \cite{GGMRR_21,GRR_22}, extended to the nonautonomous case in \cite{Riv_23}, and studied from the distributional and geometric viewpoint in \cite{LRS_25}. More recently, symmetries and dissipation laws in the \(k\)-contact setting have also been further developed in \cite{RSS_24}. Altogether, these works show that \(k\)-contact geometry provides a natural covariant extension of contact geometry to dissipative field theories. It is also worth mentioning that \(k\)-contact geometry has recently been applied to control systems, integrable systems, Lie systems, and first-order differential equations \cite{SF_25,LRS_25,deLucasRivasSobczak2026}, while also motivating several differential-geometric problems \cite{LRS_25,Riv_22,SF_25}.

The purpose of the present article is to show that the \(k\)-contact Hamilton--De Donder--Weyl formalism is not merely a conceptual extension of conservative field theory, but also a practical source of explicit geometric realizations for a broad class of nonconservative systems of partial differential equations. Our aim is not to give an exhaustive classification, but rather to identify robust geometric mechanisms that produce dissipative PDEs from \(k\)-contact Hamiltonian systems, to organise a substantial family of worked examples within a common framework, and to indicate further classes of models that appear accessible from the same viewpoint. In particular, the present paper substantially extends the range of applications of \(k\)-contact geometry to dissipative field theories and, to a significant extent, the corresponding geometric methods.

Two complementary geometric settings play a central rôle throughout the paper. The first one is the canonical \(k\)-contact manifold \(\bigoplus^k \T^*Q\times\mathbb R^k\), which is employed to analyse damped equations and Hamiltonians with polynomial  dependence on the dissipative variables. We recall from \cite{deLucasLangeSardonRivas2026} the regularity theory for these models and the analysis of systems of PDEs generated by Hamiltonian functions with a term with polynomial dependence on the dissipative variables. As a novelty with respect to \cite{deLucasLangeSardonRivas2026}, we analyse the nature of the systems of PDEs generated in this method. In particular, it is studied whether these systems of PDEs are hyperbolic, elliptic, or ultrahyperbolic. 

The second geometric setting consists of two-contactifications of  exact  $k$-symplectic phase spaces of particular types. This is indeed a new framework motivated by our applications. In particular some of our models are defined on a cotangent bundle whose natural exact symplectic structure admits a second exact closed two-form giving rise to a two-symplectic manifold. Then, a two-contactification of this model is developed. This second mechanism is analysed  systematically: in Section~\ref{sec:tools_apps} we prove a splitting result for the HdDW equations on \(k\)-contactifications of exact \(k\)-symplectic manifolds, introduce {\it adapted two-contact phase spaces of cotangent type}, formulate a partial regularity condition with respect to the spatial momenta, and show how this allows one to eliminate those momenta and recover second-order PDEs from first-order two-contact Hamiltonian systems. 

Section \ref{Sec:Dissipation} develops dissipation laws and explains how infinitesimal dynamical symmetries generate geometrically natural dissipated quantities recalling previous results in \cite{CCM_18,Riv_22}. 

All our techniques and tools are then applied in Section~\ref{sec:examples} to a wide family of systems of PDEs with physical applications. On the canonical side, we obtain \(k\)-contact realizations for damped Klein--Gordon-type equations \cite{WangCheng2005}, damped sine--Gordon equations \cite{McLaughlinScott1978,CostabileEtAl1978,CuevasMaraverKevrekidisWilliams2014}, and a damped \(\phi^4\) field on \(\mathbb R^{3,1}\) \cite{QuinteroSanchezMertens2001,LizunovaVanWezel2021}, while Section~\ref{sec:2} also contains a damped double sine--Gordon example illustrating the general canonical mechanism \cite{Popov2005,TakacsWagner2006}. By means of the adapted two-contact phase spaces formalism, we construct   descriptions for Allen--Cahn and reaction--diffusion equations \cite{AllenCahn1979,Fife1979}, a family of generalized Burgers-type equations extending the classical viscous Burgers equation \cite{Burgers1948}, the porous medium equation with linear absorption  \cite{Barenblatt1952,Aronson1986,Vazquez2007,BandleNanbuStakgold1998,Kwak1998}, the Fisher--KPP equation with linear loss \cite{Fisher1937,Kolmogorov1937,Murray2002,OrugantiShiShivaji2002,KurataShi2008,Wang2016}, the complex Ginzburg--Landau equation \cite{KuramotoTsuzuki1976,AransonKramer2002,CrossHohenberg1993,LugiatoLefever1987}, damped nonlinear Schr\"odinger-type systems \cite{SulemSulem1999}, and the FitzHugh--Nagumo model \cite{FitzHugh1961,Nagumo1962,Hastings1982FHN,ErmentroutTerman2010}. These examples, summarised in Table~\ref{tab:kcontact_examples}, show that the same geometric formalism can accommodate dissipative equations, nonlinear diffusion, reaction--diffusion dynamics, complex amplitude equations, and excitable-media systems within a unified Hamiltonian picture. When applicable, we also discuss the analytic type of the resulting systems of second-order PDEs.

The contribution of the paper is therefore threefold. First, it provides a substantial collection of explicit dissipative PDEs admitting genuine \(k\)-contact Hamilton--De Donder--Weyl formulations. Second, it develops a geometric toolbox for constructing such realizations, especially in the reduced exact two-symplectic setting, where the distinction between field variables and spatial momenta becomes intrinsic. Third, it shows that the theory is flexible enough to suggest many further applications beyond the worked examples: the additional models collected in Table \ref{tab:kcontact_applications_II} indicate natural candidate equations and Hamiltonians for future developments, ranging from further damped scalar field equations to thermoelastic, sigma-model, relaxation, and multi-sector thermodynamic systems. It is worth noting that Hamiltonians with polynomial dependence on the dissipative variables are also studied, clarifying their geometric meaning and their use in the construction of dissipative PDEs. In particular, quadratic terms in the dissipative variables are used to describe systems of PDEs with coefficients depending on independent variables.

At the same time, the paper also clarifies the present scope of the theory. The first-order \(k\)-contact formalism developed here works particularly well for equations that can be encoded either directly on canonical \(k\)-contact manifolds or on contactifications of reduced exact \(k\)-symplectic phase spaces with suitable regularity properties. This viewpoint already covers a broad range of dissipative models, but it also points towards natural extensions, including higher-order \(k\)-contact field theories, singular and constrained formulations, symmetry reduction, Hamilton--Jacobi methods, and deeper relations with nonconservative continuum theories \cite{deLucasLangeSardonRivas2026}.

The paper is organised as follows. Section~\ref{sec:2} reviews the basic notions of \(k\)-contact geometry, including \(k\)-contact Hamiltonian systems, and regularity for relevant classes of Hamiltonians. Section~\ref{sec:tools_apps} develops the general tools used later on, including splitting results for HdDW equations on \(k\)-contactifications, adapted reduced two-contact phase spaces, partial regularity, and dissipation laws. Section~\ref{sec:examples} contains the main applications. Finally, Section~\ref{sec:conclusions} summarises the conclusions and discusses further directions suggested by the present results.



\begin{table}[ht]
\centering
\footnotesize
\setlength{\tabcolsep}{4pt}
\renewcommand{\arraystretch}{1.18}
\begin{tabularx}{\textwidth}{|p{2.7cm}|Y|Y|}
\hline
\textbf{Model} & \textbf{Field equation} & \textbf{Hamiltonian / geometric realization} \\
\hline
Damped wave & \(u_{tt}-c^2u_{xx}+k u_t=0\) & On \(\mathcal J_{\mathbb R,2}\): \(h=\frac{1}{2\rho}(p^t)^2-\frac{1}{2\tau}(p^x)^2+k z^t\), with \(c^2=\tau/\rho\). \\
\hline
Damped Klein--Gordon-type model & \(u_{tt}-c^2u_{xx}+z^t u_t+\epsilon u=0\) & On \(\mathcal{J}_{:\mathbb{R},2}\): \(h=\frac12\bigl((p^t)^2-c^{-2}(p^x)^2\bigr)+\frac12(z^t)^2+\frac{\epsilon}{2}u^2\). \\
\hline
Allen--Cahn type / reaction--diffusion & \(u_t-\nu u_{xx}+V'(u)+\lambda u=0\) & On the contactification of a reduced exact \(2\)-symplectic phase space: \(h_{\mathrm{RD}}=-\nu^{-1}p^xq^x+v\bigl(V'(u)+\lambda u\bigr)\). \\
\hline
Generalized Burgers family & \(u_t-\partial_x\!\bigl(D(u)u_x\bigr)+B(u)u_x+C(u)=0\) & Two-contactification of an exact two-symplectic phase space: \(h=-D(u)^{-1}p_u^x p_v^x-\frac{B(u)}{D(u)}\,z^x+v\,C(u)\). It includes viscous Burgers for \(D(u)=\nu\), \(B(u)=u\), \(C(u)=0\). \\
\hline
Porous medium with linear absorption& \(u_t=(u^m)_{xx}-bu\), \(m>1\) & Two-contactification of an exact two-symplectic phase space: \(h_{\mathrm{PME}}=-(m u^{m-1})^{-1}p^x q^x+2bz^t\). This realization is local on the region where \(u\neq0\). \\
\hline
Complex Ginzburg--Landau & \(\psi_t=(1+i\alpha)\psi_{xx}+(1-\gamma)\psi-(1+i\beta)|\psi|^2\psi\) & In real form, on the two-contactification of an exact two-symplectic phase space: \(h_{\mathrm{CGL}}=-\frac{p_a^x(q_a^x+\alpha q_b^x)+p_b^x(-\alpha q_a^x+q_b^x)}{1+\alpha^2}-c\,R_a-d\,R_b+2\gamma z^t\). \\
\hline
Damped nonlinear Schr\"odinger & \(i\psi_t+\psi_{xx}+|\psi|^2\psi+i\gamma\psi=0\) & In real form, on the two-contactification of an exact two-symplectic phase space: \(h_{\mathrm{dNLS}}=-p_a^x q_b^x+p_b^x q_a^x-c\,R_a(a,b)-d\,R_b(a,b)\). \\
\hline
Fisher--KPP with linear loss & \(u_t=\mu u_{xx}+ru(1-u)-\lambda u\) & Two-contactification of an exact two-symplectic phase space: \(h_{\mathrm{FKPP}}=-\mu^{-1}p^xq^x+v\bigl(-ru(1-u)+\lambda u\bigr)\). \\
\hline
Damped \(\phi^4\) field & \(\phi_{tt}-\Delta\phi+m^2\phi+g\phi^3+\lambda\phi_t=0\) & On \(\mathcal J_{\mathbb R,4}\): \(h_{\phi^4}=\frac12\bigl((p^0)^2-\sum_{i=1}^3(p^i)^2\bigr)+\frac{m^2}{2}\phi^2+\frac{g}{4}\phi^4+\lambda z^0\). \\
\hline
Damped sine--Gordon & \(u_{tt}-c^2u_{xx}+\sin u+\lambda u_t=0\) & On \(\mathcal J_{\mathbb R,2}\): \(h_{\mathrm{sG}}=\frac12\bigl((p^t)^2-c^{-2}(p^x)^2\bigr)+(1-\cos u)+\lambda z^t\). \\
\hline
FitzHugh--Nagumo & \(u_t=D_u u_{xx}+f(u)-v+I\), \(v_t=D_v v_{xx}+\varepsilon(u-a v)\) & Two-contactification of an exact two-symplectic phase space: \(h_{\mathrm{FHN}}=-D_u^{-1}p_u^xq_u^x-D_v^{-1}p_v^xq_v^x-r(f(u)-v+I)-s\,\varepsilon(u-a v)\). \\
\hline
\end{tabularx}
\caption{Examples developed in the paper. Canonical \(k\)-contact models coexist with realizations on contactifications of reduced exact \(k\)-symplectic phase spaces, depending on the structure of the equation under consideration.}
\label{tab:kcontact_examples}
\end{table}
\section{{\it k}-contact manifolds and {\it k}-contact Hamiltonian systems}
\label{sec:2}

Let us describe the fundamental notions, results, and assumptions used hereafter in this work. Recall that $\mathbb{R}^k$ admits a canonical basis $\{e_1,\ldots,e_k\}$. Direct sums of vector bundles are considered to be Whitney sums and structures are smooth unless otherwise stated.  Manifolds are always finite-dimensional.  

After clarifying some standard notions on distributions and $k$-vector fields, we recall the definition of $k$-contact manifolds and $k$-contact Hamiltonian systems. Two main examples of $k$-contact manifolds are studied in this work: first-order jet manifolds and $k$-contactifications of exact symplectic manifolds. 
After explaining these notions, our work turns to  analysing how our formalism applies to Hamiltonian functions with an affine and polynomial dependence on the so-called dissipative variables and how such dependences yield second-order partial differential equations provided regularity conditions are satisfied, which is relevant for applications.

\subsection{On {\it k}-vector fields and related notions}

Let us recall some basic definitions and concepts, including $k$-contact vector fields, $k$-contact geometry, and other related notions to be used hereafter \cite{LSV_15,MRSV_15,Riv_22,LRS_25,deLucasRivasSobczak2026}.

    Assume $M$ to be an $m$-dimensional  manifold and let $\T M$ be its tangent bundle.
    \begin{itemize}
        \item A \textit{distribution} on $M$ attaches each $x\in M$ with a vector subspace  $D_x = D\,\cap\T_xM$ of $\T_xM$. We call $\dim D_x$ the {\it rank} of $D$ at $x\in M$.
        \item A distribution $D$ is \textit{smooth} if there exists a neighbourhood $U$ around every $x\in M$ and a family of vector fields $X_1, \ldots, X_r:U\rightarrow \T U$ such that $D_{x'}=\langle X_1(x'), \ldots, X_r(x')\rangle$ for every $x'\in U$. Note that $r$ does not need to match the dimension of $D_x$. 
        \item A distribution $D$ is \textit{regular} if it is smooth and of constant rank. 
    \end{itemize}

If $D$ has constant rank $k$, we write $\rk D = k$. If $D$ has constant corank $p$, we denote it by $\cork D = p$.  
The space of vector fields on $M$ taking values in a distribution $D\subset \T M$ is denoted by $\Gamma(D)$. Note that we use the term `distribution' instead of `generalised distribution', as in part of the standard literature, to simplify the notation and because it does not lead to any misunderstanding.

Let us write $\bigoplus^k\T M = \T M\oplus\overset{(k)}{\dotsb}\oplus_M \T M$ for the Whitney sum of $k$ copies of $\T M$. Define also the natural projections
\begin{equation*}
    \pr^\alpha\colon\bigoplus\nolimits^k\T M\to\T M\,, \qquad \pr_M\colon\bigoplus\nolimits^k\T M\to M\,, \qquad \alpha=1,\ldots, k\, ,
\end{equation*}
where $\pr^\alpha$ stands for the projection onto the $\alpha$-th component of the Whitney sum. A {\it $k$-vector field} on $M$ is a section ${\bf X}\colon M\to\bigoplus^k\T M$ of $\pr_M$, i.e. a map making the diagram
\begin{equation}
\label{diag:projection}
    \begin{tikzcd}[row sep=huge, column sep=huge]
        & \bigoplus^k\T M \arrow[d, "\pr^\alpha"]\\
        M \arrow[r, "X_\alpha"] \arrow[ur, "{\bf X}"] & \T M
    \end{tikzcd}
\end{equation}
commutative. Throughout the paper, $\mathfrak{X}^k(M)$ denotes the space of $k$-vector fields on $M$.

Diagram \eqref{diag:projection} allows us to consider a $k$-vector field ${\bf X}\in\mathfrak{X}^k(M)$ as a family of vector fields $X_1,\dotsc,X_k\in\mathfrak{X}(M)$ given by $X_\alpha = \pr^\alpha\circ{\bf X}$ with $\alpha=1,\ldots,k$. Taking this into account, one can denote ${\bf X} = (X_1,\dotsc, X_k)$. 
 \label{dfn:first-prolongation-k-tangent-bundle}
    Given a map $\phi\colon U\subset\R^k\to M$, its {\it first prolongation} to $\bigoplus^k\T M$ is the map $\phi'\colon U\subset\R^k\to\bigoplus^k\T M$ 
    defined by
    $$ \phi'(t) = \left( \phi(t); \T_t\phi\left( \parder{}{t^1}\bigg\vert_t \right),\dotsc,\T_t\phi\left( \parder{}{t^k}\bigg\vert_t \right) \right)\,, \qquad t\in\mathbb{R}^k, $$
    where $t^1,\dotsc,t^k$ are the canonical Cartesian coordinates on $\R^k$.

    An {\it integral section} of a $k$-vector field ${\bf X} = (X_1,\dotsc,X_k)\in\mathfrak{X}^k(M)$ is a map $\phi\colon U\subset\R^k\to M$ such that
    $ \phi' = {\bf X}\circ\phi\,, $
    namely $\T\phi \left(\parder{}{t^\alpha}\right) = X_\alpha\circ\phi$ for $\alpha=1,\ldots,k$. A $k$-vector field ${\bf X}\in\mathfrak{X}^k(M)$ is {\it integrable} if and only if every point of $M$ is in the image of an integral section of ${\bf X}$.

  Let ${\bf X} = (X_1,\ldots, X_k)$ be a $k$-vector field  on $M$ with local expression $ X_\alpha = \sum_{i=1}^nX_\alpha^i\parder{}{x^i}\,$ for $\alpha=1,\ldots,k$. 
Then, $\phi\colon U\subset\R^k\to M$ is an integral section of ${\bf X}$ if and only if its coordinates, $\phi^1,\ldots,\phi^n$, form a solution of the system of partial differential equations
\begin{equation}
\label{eq:PDEs}
\parder{\phi^i}{t^\alpha} = X_\alpha^i\circ\phi\,,\qquad i=1,\ldots,n\,,\qquad \alpha=1,\ldots,k\,.
\end{equation}

Then, ${\bf X}$ is integrable if, and only if, $[X_\alpha,X_\beta] = 0$ for $\alpha,\beta=1,\ldots,k$. These are precisely the necessary and sufficient conditions for the integrability of the system of PDEs \eqref{eq:PDEs} (see \cite{OLV_86} for details).

Each $\bm\theta\in \Omega^\ell(M,\R^k)$ can be represented uniquely as $\bm\theta=\sum_{\alpha=1}^k\theta^\alpha\otimes e_\alpha$ for some differential $\ell$-forms $\theta^1,\ldots,\theta^k \in \Omega^\ell(M)$. 
Then, $\bm\theta\in \Omega^\ell(M,\R^k)$ is {\it nondegenerate} if $\ker\bm\theta:=\bigcap^k_{\alpha=1}\ker\theta^\alpha=0$.
The contractions of a $k$-vector field ${\bm Z}$ and a vector field $Z\in \mathfrak{X}(M)$ on $M$ with $\bm \theta\in \Omega^\ell(M,\mathbb{R}^k)$ are given, respectively, by
$$
\iota_{\bm Z}\bm \theta=\sum_{\alpha=1}^k\iota_{Z_\alpha}\theta^\alpha \in \Omega^{\ell-1}(M),\qquad
\iota_{Z}\bm \theta=\sum_{\alpha=1}^k\iota_{Z}\theta^\alpha\otimes e_\alpha\in \Omega^{\ell-1}(M,\mathbb{R}^k).
$$

Let us finally define the fibred differential. 
\begin{definition}
Let 
\(
\alpha \colon Q \times \mathbb{R}^k \longrightarrow 
\bigwedge\nolimits^k \cT Q \times \mathbb{R}^k
\)
be a section and fix $z\in\mathbb{R}^k$.
Define
\begin{align*}\textstyle
\alpha_z \colon Q &\longrightarrow \bigwedge\nolimits^k \cT Q \\
x & \longmapsto \mathrm{pr}_{\Lambda^k \cT Q}(\alpha(x,z))\,,
\end{align*}
where $\mathrm{pr}_{\Lambda^k \cT Q}$ denotes the canonical projection from $\bigwedge\nolimits^k \cT Q \times \mathbb{R}^k$ onto $\Lambda^k\T^*Q$.
The \emph{exterior derivative of $\alpha$ over the bundle $Q\times \mathbb{R}^k$ onto $Q$} is the section
\(
\dd_Q \alpha \colon Q \times \mathbb{R}^k \longrightarrow
\bigwedge\nolimits^{k+1}\cT Q \times \mathbb{R}^k
\)
defined by
\(
\dd_Q \alpha(x,z) = (\dd \alpha_z(x),\, z)\,.
\)
\end{definition}
\subsection{\texorpdfstring{$k$}--contact geometry}

Let us present the basic notions of $k$-contact geometry to be used in this work \cite{LRS_25,GGMRR_20,GGMRR_21}. 

\begin{definition}\label{def:k-form}
A \textit{$k$-contact form on an open subset $U \subset M$} is a differential one-form on $U$ taking values in $\mathbb{R}^k$, that is $\bm{\eta} \in \Omega^1(U,\mathbb{R}^k)$, such that
 $\ker \bm{\eta} \subset \T U$ is a regular non-zero distribution of corank $k$ and $\ker \bm{\eta} \oplus \ker \dd \bm{\eta} = \T U$.

If $\bm{\eta}$ is a $k$-contact form defined on $M$, the pair $(M,\bm{\eta})$ is called a \textit{co-oriented $k$-contact manifold}, and $\ker \dd \bm{\eta}$ is called the \textit{Reeb distribution} of $(M,\bm{\eta})$. Moreover, if $\dim M = n + nk + k$ for some $n,k \in \mathbb{N}$ and there exists an integrable distribution $\mathcal{V} \subset \ker \bm{\eta}$ with $\rk \mathcal{V} = nk$, then $(M,\bm{\eta},\mathcal{V})$ is called a \textit{polarised co-oriented $k$-contact manifold}. The distribution $\mathcal{V}$ is referred to as a \textit{polarisation} of $(M,\bm{\eta})$.
\end{definition} 

According to the above formalism, a one-contact form is a contact form (cf. \cite{Riv_22}). Nevertheless, Definition~\ref{def:k-form} also requires $\ker \bm{\eta}\neq 0$ to avoid considering as a contact manifold the case $k=1$ with $\dim M=1$. In this instance, $\ker \eta=0$ is sometimes said to be {\it maximally non-integrable} \cite{GG_22}, although it is an integrable distribution. Moreover, allowing this case involves dealing with several technical nuances in $k$-contact geometry, which is undesirable. Despite this,  $\ker \bm{\eta}\neq 0$ is a mere technical condition that excludes an instance of no real relevance in contact or $k$-contact geometry. 

\begin{theorem}
\label{thm:k-contact-Reeb}
Let $\big(M,\bm\eta = \sum_{\alpha=1}^k \eta^\alpha \otimes e_\alpha\big)$ be a co-oriented $k$-contact manifold. There exists a unique family of $k$-vector fields $R_1,\ldots,R_k \in \mathfrak{X}(M)$ such that
\begin{equation}\label{eq:k-contact-Reeb} 
\inn{R_\alpha} \eta^\beta = \delta_\alpha^\beta\,,\qquad 
\inn{R_\alpha} \dd \eta^\beta = 0\,, \qquad \alpha,\beta = 1,\ldots,k\,.
\end{equation}
The vector fields $R_1,\ldots,R_k$ commute with each other and form a basis of the Reeb distribution 
\(
\ker \dd \bm\eta.
\)
\end{theorem}

Since $\ker \dd \bm\eta$ is the intersection of the kernels of closed two-forms and has constant rank $k$ by assumption, it is integrable. 

\begin{definition}
Given a $k$-contact manifold $(M,\bm\eta)$, the \textit{Reeb $k$-vector field} of $(M,\bm\eta)$ is the integrable $k$-vector field
\(
\mathbf{R} = (R_1,\ldots,R_k)
\) 
on $M$ whose components are described in Theorem~\ref{thm:k-contact-Reeb}. The vector fields $R_1,\ldots,R_k$ are called the \textit{Reeb vector fields} of the co-oriented $k$-contact manifold $(M,\bm\eta)$.
\end{definition}

The following two examples describe the main $k$-contact manifolds to be used in this work. The first is also the canonical local model for polarised $k$-contact manifolds as they are all locally of this type \cite{LRS_25}. In particular, it is related to first-order jet bundles and many examples in physics \cite{LRS_25}. 

\begin{example} 
\label{ex:canonical-k-contact-structure}
\rm
The manifold $\mathcal{J}_{Q,k} = \bigoplus^k \cT Q \times \mathbb{R}^k$ carries a canonical $k$-contact manifold structure defined by the $k$-contact form
\(
\bm\eta_{Q,k} = \sum_{\alpha=1}^k (\dd z^\alpha - \theta^\alpha) \otimes e_\alpha.
\)
Here, $\{z^1,\ldots,z^k\}$ are linear coordinates on $\mathbb{R}^k$, pulled back to $\mathcal{J}_{Q,k}$ via its canonical projection onto $\mathbb{R}^k$. These coordinates are hereafter called {\it dissipative variables} due to their relation to dissipation in field theories \cite{Riv_22}. Moreover, each $\theta^\alpha$ is the pull-back to $\mathcal{J}_{Q,k}$ of the Liouville one-form $\theta$ on the $\alpha$-th copy of the cotangent bundle $\cT Q$ via the natural projections $\mathcal{J}_{Q,k}\to \bigoplus^k\T^*Q\stackrel{{\rm pr}_\alpha}{\rightarrow} \cT Q$. Additionally, the canonical projection $\mathcal{J}_{Q,k}\rightarrow  Q\times \mathbb{R}^k$ gives rise to a vertical distribution $\mathcal{V}\subset \T\mathcal{J}_{Q,k}$ of rank $k \cdot \dim Q$ contained in $\ker \bm\eta_{Q,k}$. Therefore, $\left(\mathcal{J}_{Q,k}, \bm\eta_{Q,k}, \mathcal{V}\right)$ is a polarised co-oriented $k$-contact manifold.

Any choice of local coordinates $\{q^1,\ldots,q^n\}$ on $Q$, together with linear coordinates on $\mathbb{R}^k$, induces a natural coordinate system $\{q^i, p_i^\alpha, z^\alpha\}$, with $\alpha = 1,\ldots,k$,  on $\mathcal{J}_{Q,k}$. In these coordinates,
\[
{\bm \eta}_{Q,k}= \sum_{\alpha=1}^k\left(\dd z^\alpha - \sum_{i=1}^n p_i^\alpha \dd q^i\right)\otimes e_\alpha,\quad
\dd {\bm \eta}_{Q,k}= \sum_{\alpha=1}^k\left(\sum_{i=1}^n\dd q^i \wedge \dd p_i^\alpha\right)\otimes e_\alpha,\quad
R_\beta = \parder{}{z^\beta},
\]
for \(\beta = 1,\ldots,k\), and
\[
\ker \bm\eta_{Q,k} =
\left\langle
\parder{}{p_i^\alpha},
\parder{}{q^i} + \sum_{\alpha=1}^k p_i^\alpha \parder{}{z^\alpha}
\right\rangle_{ i = 1,\ldots,n}\!\!\!\!\!\!\!\!\!\!\!\!\!\!\!\!,\!\!\!\!\quad
\ker \dd \bm\eta_{Q,k} =
\left\langle
R_1,\ldots,R_k
\right\rangle,
\!
\mathcal{V}_{Q,k}=\left\langle\parder{}{p^\alpha_i}\right\rangle_{\substack{i=1,\ldots,n\\\alpha=1,\ldots,k}} .
\]
\end{example}

\begin{example}\label{ex:contactification-k-symplectic-manifold}
Let us study the $k$-contactification of an exact $k$-symplectic manifold. Recall that a $k$-symplectic manifold is a pair $(P,\bm \omega)$, where $\bm \omega\in \Omega^2(P,\mathbb{R}^k)$ is a closed two-form taking values in $\mathbb{R}^k$ with $\ker \bm\omega=0$. More specifically, $(P,\bm \omega)$ is {\it exact} if $(P,\bm\omega = \sum_{\alpha=1}^k \omega^\alpha \otimes e_\alpha)$ with $\bm \omega = -\dd \bm\theta=-\sum_{\alpha=1}^k\dd \theta^\alpha\otimes e_\alpha$. Then, the $k$-contactification of $(P,-\dd\bm\theta)$ is the product manifold $M = P \times \mathbb{R}^k$ along with the 
$\mathbb{R}^k$-valued one-form
\[
\bm\eta_{\bm \theta} = \sum_{\alpha=1}^k ( \dd z^\alpha - \theta_M^\alpha)\otimes e_\alpha \in \Omega^1(M,\mathbb{R}^k),
\]
where $\{z^1,\ldots,z^k\}$ denote the pull-back  via the natural projection $M\rightarrow \mathbb{R}^k$ of linear coordinates on $\mathbb{R}^k$, and each $\theta_M^\alpha$ is the pull-back of $\theta^\alpha$ from $P$ to $M$ relative to the canonical projection $\rho:M\rightarrow P$. Then, $(M,\bm\eta_{\bm \theta})$ becomes a co-oriented $k$-contact manifold. Indeed, $\ker \bm\eta_{\bm \theta}$ has corank $k$ and is non-zero, while $\dd \bm\eta_{\bm \theta} = -\dd \bm\theta_M$ for $\bm \theta_M=\sum_{\alpha=1}^k\theta^\alpha_M\otimes e_\alpha$. Then,
\[
\ker \dd \bm\eta_{\bm\theta} = \left\langle \parder{}{z^1},\ldots,\parder{}{z^k} \right\rangle
\]
has rank $k$, since $(P,\bm\omega)$ is $k$-symplectic. Hence, $\bm\eta_{\bm \theta}$ defines a $k$-contact form on $M$.

The  $k$-contact form $\bm\eta_{Q,k}$ described in Example~\ref{ex:canonical-k-contact-structure} is  the $k$-contactification of the exact $k$-symplectic manifold $(P = \bigoplus^k \cT Q, \bm\omega_{Q,k}=\sum_{\alpha=1}^k \omega_{Q}^\alpha \otimes e_\alpha)$, where $\omega_{Q}^\alpha$ denotes the pull-back to $P$ of the canonical symplectic form from the $\alpha$-th copy of $\cT Q$. Since $\T^*Q$ is an exact symplectic manifold with a canonical Liouville form $\theta$, then $\bm\omega_{Q,k}=-\dd \bm\theta_{Q,k}$ where $\bm\theta_{Q,k} = \sum_{\alpha=1}^k \theta_Q^\alpha \otimes e_\alpha$ and $\theta^\alpha_Q$ is the pull-back to $P$ of the  Liouville form on the $\alpha$-copy of $\T^*Q$ in $P$.
\end{example}

\begin{theorem}[$k$-contact Darboux Theorem \cite{GGMRR_20}]
\label{thm:kcontDarboux}
Let $(M,\bm\eta,\mathcal{V})$ be a polarised co-oriented $k$-contact manifold. Around every point of $M$, there exist local coordinates $\{q^i,p_i^\alpha,z^\alpha\}$, with $1 \leq i \leq n$ and $1 \leq \alpha \leq k$, such that
\[
\bm \eta = \sum_{\alpha=1}^k \left(\dd z^\alpha - \sum_{i=1}^np_i^\alpha \dd q^i\right)\otimes e_\alpha,
\,\,
\ker \dd \bm\eta = \left\langle   \parder{}{z^1} ,\ldots,  \parder{}{z^k} \right\rangle,
 \,\,
\mathcal{V} = \left\langle \parder{}{p_i^\alpha} \right\rangle_{\substack{i=1,\ldots,n\\\alpha=1,\ldots,k}} .
\]
These coordinates are called \textit{Darboux coordinates} of $(M,\bm\eta,\mathcal{V})$. Note that the Reeb vector fields read $R_\alpha=\partial/\partial z^\alpha$ for $\alpha=1,\ldots,k$. 
\end{theorem}

Theorem~\ref{thm:kcontDarboux} allows us to regard the manifold introduced in Example~\ref{ex:canonical-k-contact-structure} as the canonical local model of polarised co-oriented $k$-contact manifolds (see \cite{LRS_25} for details).

\subsection{On {\it k}-contact Hamiltonian systems}
\label{sub:k-contact-Hamiltonian-systems}

After introducing the geometric framework of $k$-contact geometry, we now address its Hamiltonian formulation of field theories. Special attention is paid to the relation of Hamilton--De Donder-Weyl (HdDW) equations in $k$-contact polarised manifolds and the relation of the described equations with the theory of partial differential equations.

\begin{definition}
The {\it $k$-contact Hamilton--De Donder--Weyl equations} for a $k$-vector field
${\bf X} = (X_1,\ldots,X_k) \in \mathfrak{X}^k(M)$ associated with $h\in \Cinfty(M)$ and $(M,\bm \eta)$ are
\begin{equation}\label{eq:k-contact-HdDW-fields1}
\iota_{\bm X}\dd\bm\eta=\sum_{\alpha=1}^k \inn{X_\alpha} \dd \eta^\alpha
= \dd h - \sum_{\alpha=1}^k (\Lie_{R_\alpha} h)\, \eta^\alpha,\qquad
\iota_{\bm X}\bm\eta=\sum_{\alpha=1}^k \inn{X_\alpha} \eta^\alpha
= -h .
\end{equation}
A $k$-vector field ${\bm X}$ satisfying these equations for some $h$ is called an
\textit{${\bm \eta}$-Hamiltonian $k$-vector field}. Meanwhile,  $(M,\bm \eta,h)$ is called an {\it $\bm\eta$-Hamiltonian system} and $h \in \Cinfty(M)$ is its \textit{Hamiltonian function}.
\end{definition}

The space of $\bm \eta$-Hamiltonian $k$-vector fields is a vector space. From now on, $(M,\bm\eta)$ always stands for a $k$-contact manifold, and $(M,\bm\eta,h)$ is an ${\bm \eta}$-Hamiltonian system. Every $(M,\bm \eta,h)$ induces an ${\bm \eta}$-Hamiltonian $k$-vector field \cite{HLM_26,Riv_22}. Indeed, one has the following result (see \cite[Theorem 3.6]{HLM_26} for more details).

\begin{proposition}\label{prop:k-contact-HdDW-have-solutions}
The $k$-contact Hamilton--De Donder--Weyl equations for a $k$-vector field ${\bf X}$ given by
\eqref{eq:k-contact-HdDW-fields1} admit solutions for every $h\in \Cinfty(M)$. These solutions are not unique only if $k>1$.
\end{proposition}

Let ${\bf X} = (X_1,\ldots,X_k) \in \mathfrak{X}^k(M)$ be a $k$-vector field with local expression in Darboux coordinates
\[
X_\alpha
=
\sum_{i=1}^n (X_\alpha)^i \parder{}{q^i}
\;+\;
\sum_{\beta=1}^k \sum_{i=1}^n (X_\alpha)^\beta_i \parder{}{p_i^\beta}
\;+\;
\sum_{\beta=1}^k (X_\alpha)^\beta \parder{}{z^\beta}\,,
\qquad \alpha=1,\ldots,k\,.
\]
Note that the existence of Darboux coordinates for some $(M,\bm\eta)$, namely putting $\bm\eta$ locally of the form ${\bm \eta}=\sum_{\alpha=1}^k (\dd z^\alpha -\sum_{i=1}^np^\alpha_i\dd q^i)\otimes e_\alpha$, yields that $(M,\bm\eta)$ induces, locally, a polarised $k$-contact manifold. Then, equations \eqref{eq:k-contact-HdDW-fields1} imply
\begin{equation}\label{eq:k-contact-HdDW-fields-Darboux-coordinates1}
\begin{dcases}
(X_\beta)^i = \parder{h}{p_i^\beta}\,,\\
\sum_{\alpha=1}^k (X_\alpha)^\alpha_i
= -\left( \parder{h}{q^i} + \sum_{\alpha=1}^k p_i^\alpha \parder{h}{z^\alpha} \right)\,,\\
\sum_{\alpha=1}^k (X_\alpha)^\alpha
= \sum_{\alpha=1}^k \sum_{j=1}^n p_j^\alpha \parder{h}{p_j^\alpha} - h\,,
\end{dcases}\qquad i=1,\ldots,n,\qquad \beta=1,\ldots,k.
\end{equation}

\begin{definition}\label{dfn:k-contact-Hamiltonian-system}
Given a map $\psi \colon D \subset \mathbb{R}^k \to M$, the \textit{$k$-contact Hamilton--De Donder--Weyl equations} related to the $\bm \eta$-Hamiltonian system $(M,\bm \eta,h)$ are
\begin{equation}\label{eq:k-contact-HdDW}
\sum_{\alpha=1}^k \inn{\psi'_\alpha} \dd \eta^\alpha
= \left( \dd h - \sum_{\alpha=1}^k (\Lie_{R_\alpha} h)\, \eta^\alpha \right) \circ \psi\,,\qquad
\sum_{\alpha=1}^k \inn{\psi'_\alpha} \eta^\alpha
= - h \circ \psi\,.
\end{equation}
\end{definition}

It is worth stressing that a solution for the $k$-contact HdDW equation \eqref{eq:k-contact-HdDW} gives rise to a $k$-contact vector field on ${\rm Im}\,\psi$, which is, locally, an embedded submanifold diffeomorphic to $D$.

In \textit{Darboux coordinates} for a polarised $k$-contact manifold $(M,\bm \eta,\mathcal{V})$, a section $\psi:\mathbb{R}^k\rightarrow M$ takes the form $\psi(t) = (q^i(t),p_i^\alpha(t), z^\alpha(t))$, while equations \eqref{eq:k-contact-HdDW} read
\begin{equation}\label{eq:k-contact-HdDW-Darboux-coordinates}
\begin{dcases}
\parder{q^i}{t^\beta} = \parder{h}{p_i^\beta} \circ \psi\,,\\
\sum_{\alpha=1}^k \parder{p_i^\alpha}{t^\alpha}
= -\left( \parder{h}{q^i} + \sum_{\alpha=1}^k p_i^\alpha \parder{h}{z^\alpha} \right) \circ \psi\,,\\
\sum_{\alpha=1}^k \parder{z^\alpha}{t^\alpha}
= \left( \sum_{\alpha=1}^k \sum_{j=1}^n p_j^\alpha \parder{h}{p_j^\alpha} - h \right) \circ \psi\,,
\end{dcases}
\end{equation}
for $i=1,\ldots,n$ and $\beta=1,\ldots,k$.

The second line of equations in \eqref{eq:k-contact-HdDW-Darboux-coordinates} are called the \textit{balance equations of the polymomenta,} since they are related to the balance of momenta, their conservation laws, and their dissipation in field theories \cite{Riv_22}. Meanwhile, the third line of equations in \eqref{eq:k-contact-HdDW-Darboux-coordinates} is called the \textit{balance equation of the dissipative variables}, since it is related to the balance of the dissipative variables and their dissipation in field theories \cite{Riv_22}. The first line of equations in \eqref{eq:k-contact-HdDW-Darboux-coordinates} are the so-called  \textit{evolution equations of the fields $q^i(t)$}, since they are related to the evolution of the fields in field theories \cite{Riv_22}.

The existence of solutions of equations \eqref{eq:k-contact-HdDW-fields1}
does not imply the integrability of the associated $k$-vector fields. Indeed, as in the $k$-symplectic case, equations \eqref{eq:k-contact-HdDW} and
\eqref{eq:k-contact-HdDW-fields1} are not fully equivalent, since a $k$-vector field solving \eqref{eq:k-contact-HdDW-fields1} does not need to give rise to a solution
of \eqref{eq:k-contact-HdDW}. For instance, this happens when $X_1,\ldots,X_k$ do not commute between themselves at any point of $M$ (see \cite{GGMRR_20} for further details).
 
Let us illustrate our previous theory by studying a damped double sine--Gordon equation. Consider coordinates \((t,x)\) on \(\mathbb R^2\) and let \(Q=\mathbb R\). The phase space is \(\mathcal J_{\mathbb R,2}=(\T^*\mathbb R\oplus \T^*\mathbb R)\times\mathbb R^2\) with coordinates \((u,p^t,p^x,z^t,z^x)\). The variable \(u\) represents a scalar field, while \(p^t\) and \(p^x\) denote the conjugate momenta associated with the independent variables \(t\) and \(x\), respectively. The double sine--Gordon equation appears in several physical contexts, for instance in long Josephson junctions and related non-integrable field-theoretic models \cite{KasTanHataTaka2001,Popov2005,TakacsWagner2006}. Let us consider the Hamiltonian function \(h_{\mathrm{DSG}}\in C^\infty(\mathcal J_{\mathbb R,2})\) given by
\begin{equation}\label{eq:Hamiltonian-DSG}
h_{\mathrm{DSG}}(u,p^t,p^x,z^t,z^x)
=
\frac12\Bigl((p^t)^2-\frac{1}{c^2}(p^x)^2\Bigr)
+a(1-\cos u)+b(1-\cos 2u)
+\frac12(z^t)^2,
\end{equation}
where \(c>0\) and \(a,b\in\mathbb R\) are constants.

The Hamilton--De Donder--Weyl equations for the associated \(\bm\eta_{\mathbb R,2}\)-Hamiltonian two-vector field read
\[
\partial_t u=p^t,\qquad
\partial_x u=-\frac{1}{c^2}p^x,\qquad
\partial_t p^t+\partial_x p^x=-(a\sin u+2b\sin 2u+z^t p^t),
\]
together with
\[
\partial_t z^t+\partial_x z^x
=
\frac12\Bigl((p^t)^2-\frac{1}{c^2}(p^x)^2\Bigr)
-a(1-\cos u)-b(1-\cos 2u)-\frac12(z^t)^2.
\]

Using the first two equations in the third one, we obtain
\begin{equation}\label{eq:damped-DSG}
u_{tt}-c^2u_{xx}+a\sin u+2b\sin 2u+z^t u_t=0.
\end{equation}
Hence, \(h_{\mathrm{DSG}}\) yields a damped double sine--Gordon equation with effective dissipation coefficient \(z^t\). One may regard \(z^t=\lambda(t,x)\) as a prescribed space-time dependent attenuation coefficient, while \(z^x\) is chosen so that the balance equation for the dissipative variables is satisfied. In this way, \eqref{eq:damped-DSG} becomes
\[
u_{tt}-c^2u_{xx}+a\sin u+2b\sin 2u+\lambda(t,x)u_t=0.
\]
Therefore, the quadratic term \(\frac12(z^t)^2\) in the Hamiltonian allows one to encode variable dissipation directly within the two-contact formalism in a non-autonomous manner.

\subsection{HdDW-equations for relevant Hamiltonian functions and second-order PDEs}\label{Sec:SpecialTypes}

The aim of this section is to study particular types of $k$-contact Hamiltonian systems on $\mathcal{J}_{Q,k}$, their HdDW equations as well as their use to study higher-order systems of PDEs with relevant applications \cite{Riv_22}.

The following definition is a generalisation of the notion of regularity for Hamiltonian functions in the symplectic and contact cases to the $k$-contact realm developed in \cite{deLucasLangeSardonRivas2026} for studying systems of PDEs via HdDW equations. It is also a natural analogue of the notion of regularity for Lagrangian functions in the $k$-contact setting appearing in \cite{Riv_22}.

\begin{definition}\label{def:regular_hamiltonian}
We say that \((\mathcal{J}_{Q,k},\bm \eta_{Q,k},h)\) is \emph{regular} if the fibre derivative \(\mathbb Fh:\bigoplus^k\T^*Q\times\mathbb R^k\to \bigoplus^k\T Q\times\mathbb R^k\), given in canonical coordinates by \(\mathbb Fh(q^i,p_i^\alpha,z^\alpha)=(q^i,\partial h/\partial p_i^\alpha,z^\alpha)\), is a local diffeomorphism. If \(\mathbb Fh\) is a global diffeomorphism, then \((\mathcal{J}_{Q,k},\bm \eta_{Q,k},h)\) is said to be \emph{hyperregular}.
\end{definition}

In a slight abuse of notation, the terms `regular' and `hyperregular' are only used in reference to $h$ when it is understood from context that the notion is with respect to $(\mathcal{J}_{Q,k},\bm\eta_{Q,k})$. 

An analogue of the following proposition for the Lagrangian setting can be found in \cite[Proposition 8.1.6]{Riv_22}.

\begin{proposition}\label{prop:equivalent_regular_hamiltonian}
Given \((\mathcal{J}_{Q,k},\bm\eta_{Q,k}, h)\), the following conditions are equivalent: 
\begin{enumerate}
\item \(h\) is regular.
\item The Hessian matrix \(\bigl(\partial^2 h/\partial p_i^\alpha\partial p_j^\beta\bigr)\) is non-singular at every point.
\item The relations \(v_\alpha^i=\partial h/\partial p_i^\alpha\) can be solved locally for the variables \(p_i^\alpha\) as smooth functions of \((q^i,v_\alpha^i,z^\alpha)\).
\end{enumerate}
\end{proposition}

Let us now formalise for our setting a result that describes how $k$-contact geometry has been used to study higher-order systems of PDEs in the literature. In particular, we will focus on the case of Hamiltonian functions that are polynomial in the dissipative variables, which is a common situation in many models due to its simplicity and applications (cf. \cite{Riv_22,Riv_23,RSS_24,GRR_22}). The following result states that, for such Hamiltonians, every integral section of an \(\bm\eta_{Q,k}\)-Hamiltonian   $k$-vector field associated with the same Hamiltonian function projects onto a map whose coordinates satisfy the same second-order system of PDEs. It also explains certain additional properties of this case, which will be important for the applications of $k$-contact geometry to the study of PDEs in the literature and the rest of our paper.

\begin{theorem}\label{thm:linear_z_same_second_order_system}
Let \(h\in C^\infty(\mathcal J_{Q,k})\) take the form \(h(q^i,p_i^\alpha,z^\alpha)=g(q^i,p_i^\alpha)+\sum_{\alpha=1}^kA_\alpha (z^\alpha)^{\lambda_\alpha}\), where \(A_1,\dots,A_k\) are real constants, $\lambda_1,\ldots,\lambda_k$ are non-negative integers, and \(h\) is regular. Let  \(\psi(t)=(q^i(t),p_i^\alpha(t),z^\alpha(t))\) be an integral section of an \(\bm\eta_{Q,k}\)-Hamiltonian \(k\)-vector field \({\bm X}_{h}\). Then, \(\psi(t)\) projects onto a map \(q(t)=(q^1(t),\dots,q^n(t))\) satisfying a second-order system of PDEs. More precisely, if \(p_i^\alpha=P_i^\alpha(q ,v_\beta^j)\) denotes the local inverse of the fibre derivative, determined by \(v_\beta^i=\partial h/\partial p_i^\beta=\partial g/\partial p_i^\beta\), then the functions \(q^i(t)\) satisfy the second-order system of PDEs of the form
\begin{equation}\label{eq:second_order_system_linear_z}
\sum_{\alpha=1}^k \frac{\partial}{\partial t^\alpha}\Bigl(P_i^\alpha(q,v_\beta)\Bigr)
+\frac{\partial g}{\partial q^i}\Bigl(q,P_j^\gamma(q,v_\beta)\Bigr)
+\sum_{\alpha=1}^k R(z_\alpha,\lambda_\alpha) A_\alpha P_i^\alpha(q,v_\beta)=0,\qquad i=1,\ldots,n\,,
\end{equation}
where \(q=q(t)\) and \(v_\beta =\partial q /\partial t^\beta\) and $R(z ,\mu)=z ^{\mu-1}\mu$  for $\mu\in \mathbb{N}_+$ and $R(z ,0)=0$ for every $z\in \mathbb{R}$.
\end{theorem}

The system of second-order PDEs \eqref{eq:damped-DSG} is the particular case of the second-order system \eqref{eq:second_order_system_linear_z} for the Hamiltonian function \eqref{eq:Hamiltonian-DSG}, which is polynomial in the dissipative variables in the manner given in Theorem~\ref{thm:linear_z_same_second_order_system}. Hence, Theorem~\ref{thm:linear_z_same_second_order_system} ensures that every integral section of an \(\bm\eta_{Q,k}\)-Hamiltonian  $k$-vector field associated with the Hamiltonian \eqref{eq:Hamiltonian-DSG} projects onto a solution of the damped double sine-Gordon equation \eqref{eq:damped-DSG}.

It is worth noting that higher-order systems of PDEs can be written as second-order systems of PDEs by considering higher-order partial derivatives of the coordinates of solutions as new variables. In this manner, the resulting system can be described by our $k$-contact formalism.

Moreover,  Theorem \ref{thm:linear_z_same_second_order_system} can be modified to consider a Hamiltonian function of the form \(h(q^i,p_i^\alpha,z^\alpha)=g(q^i,p_i^\alpha)+\sum_{\alpha=1}^k A_\alpha(q^i)(z^\alpha)^{\lambda_\alpha}\), where \(A_1,\dots,A_k\in C^\infty(Q)\). One may also consider the particular case of Theorem \ref{thm:linear_z_same_second_order_system} given by \(h(q^i,p_i^\alpha,z^\alpha)=g(q^i,p_i^\alpha)+\sum_{\alpha=1}^k A_\alpha z^\alpha\), where \(A_1,\dots,A_k\in \mathbb{R}\). In the latter case, the second-order system of PDEs satisfied by the coordinates of the projection of integral sections of \(\bm\eta_{Q,k}\)-Hamiltonian $k$-vector fields associated with \(h\) to $Q$ is independent of the dissipative variables \(z^1,\dots,z^k\).

\begin{proposition}\label{prop:regularity_not_hyperbolicity}
Let \((J_{Q,k},\bm\eta_{Q,k},h)\) be a canonical \(k\)-contact Hamiltonian system with $\dim Q=1$, and assume  the associated  second-order system of PDEs as in Theorem~\ref{thm:linear_z_same_second_order_system}.   
If $A=\left(\frac{\partial^2 h}{\partial p^\alpha\partial p^\beta}\right)$ is definite, the resulting second-order equation is elliptic. If $A$ has signature $(k-1,1)$ or $(1,k-1)$, one obtains a hyperbolic equation. If $A$ is indefinite with at least two positive and two negative eigenvalues, the equation is of ultrahyperbolic type.  

\end{proposition} 
 
\begin{proof}
Since \(\dim Q=1\), the reconstructed second-order system reduces to a single second-order equation. Hence its analytic type is determined by the quadratic form associated with its principal symbol. By Theorem~\ref{thm:linear_z_same_second_order_system}, if \(v_\beta=\partial q/\partial t^\beta\), the reconstructed equation has the form
\[
\sum_{\alpha=1}^k \frac{\partial}{\partial t^\alpha}\Bigl(P^\alpha(q,v_\beta)\Bigr)
+\frac{\partial g}{\partial q}\Bigl(q,P^\gamma(q,v_\beta)\Bigr)
+\sum_{\alpha=1}^k R(z^\alpha,\lambda_\alpha)A_\alpha P^\alpha(q,v_\beta)=0.
\]
Therefore, the principal part of the equation, namely the one concerning second-order derivatives of the unknown variables, is
\[
\sum_{\alpha,\beta=1}^k \frac{\partial P^\alpha}{\partial v_\beta}(q,v_\gamma)\,
\frac{\partial^2 q}{\partial t^\alpha\partial t^\beta},
\]
and its principal symbol is the quadratic form
\[
\sigma(\xi)=\sum_{\alpha,\beta=1}^k \frac{\partial P^\alpha}{\partial v_\beta}(q,v_\gamma)\,\xi_\alpha\xi_\beta.
\]

Now, the regularity of \(h\) implies that the Hessian matrix
$A$
is invertible. Since \(P^\alpha\) is obtained locally by inverting the relations \(v_\alpha=\partial h/\partial p^\alpha\), the matrix
\(
\left(\frac{\partial P^\alpha}{\partial v_\beta}\right)
\)
is precisely \(A^{-1}\). Hence the principal symbol is determined by the quadratic form associated with \(A^{-1}\), and therefore by the index of \(A\), since \(A\) and \(A^{-1}\) have the same inertia.

Consequently, if \(A\) is definite, then \(\sigma\) is definite and the equation is elliptic. If \(A\) has signature \((k-1,1)\) or \((1,k-1)\), then \(\sigma\) has the same signature and the equation is hyperbolic. More generally, if \(A\) is indefinite with at least two positive and two negative eigenvalues, then the equation is of ultrahyperbolic type.

\end{proof}

\begin{remark}
    The assumption \(\dim Q=1\) is essential in the previous proposition since only in the scalar case is the analytic type determined by a single quadratic form. For higher-dimensional target manifolds, one obtains instead a matrix-valued principal symbol, and the above classification must be replaced by the corresponding analysis of the full principal symbol of the system.
\end{remark}
\section{HdDW equations for {\it k}-contactifications}
\label{sec:tools_apps}
Let us provide a general structure for the HdDW equations in $k$-contactifications of exact $k$-symplectic manifolds. This structure will be useful to study particular types of Hamiltonian functions and their associated second-order systems of PDEs in the applications of Section~4.

\begin{proposition}\label{prop:split_HdDW_contactification}
Let \((P,\bm\omega=-\dd \bm\theta)\) be an exact \(k\)-symplectic manifold and let \((M=P\times\mathbb R^k,\bm\eta_{\bm\theta})\) be its \(k\)-contactification and $\rho:M\rightarrow P$ the natural projection. Then, a section \(\psi=(\phi,z^1,\dots,z^k)\colon D\subset\mathbb R^k\to M\) satisfies the \(k\)-contact Hamilton--De Donder--Weyl equations for \(h\in C^\infty(M)\) if and only if
\begin{equation}\label{eq:split_HdDW_base}
\sum_{\alpha=1}^k \iota_{\phi'_\alpha}\omega^\alpha=\left(\dd_P h+\sum_{\alpha=1}^k \parder{h}{z^\alpha}\,\rho^*\theta^\alpha\right)\circ\psi,\qquad 
\sum_{\alpha=1}^k \parder{z^\alpha}{t^\alpha}=\left(\sum_{\alpha=1}^k\rho^*\theta^\alpha(\phi'_\alpha)-h\right)\circ\psi.
\end{equation}
In particular, the equations on the \(k\)-symplectic factor \(P\) and the equation for the dissipative variables split naturally.
\end{proposition}

\begin{proof}
Since \(R_\alpha=\partial/\partial z^\alpha\) for $\alpha=1,\ldots,k$, the \(k\)-contact Hamilton--De Donder--Weyl equations read
\[
\sum_{\alpha=1}^k \iota_{\psi'_\alpha}\dd\eta_\theta^\alpha=\left(\dd h-\sum_{\alpha=1}^k \parder{h}{z^\alpha}\eta_\theta^\alpha\right)\circ\psi,\qquad
\sum_{\alpha=1}^k \iota_{\psi'_\alpha}\eta_\theta^\alpha=-h\circ\psi.
\]
Since \(\dd{\bm \eta}_{\bm\theta}=\rho^*\bm \omega\), one has that \(\iota_{\psi'_\alpha}\dd\eta_{\theta}^\alpha=\iota_{\phi'_\alpha}\rho^*\omega^\alpha\). Moreover, since \(\eta^\alpha=\dd z^\alpha-\rho^*\theta^\alpha\), one has
\[
\dd h-\sum_{\alpha=1}^k \parder{h}{z^\alpha}\eta_\theta^\alpha=\dd_P h+\sum_{\alpha=1}^k \parder{h}{z^\alpha}\rho^*\theta^\alpha.
\]
This yields the first equation in \eqref{eq:split_HdDW_base}. On the other hand,
\[
\sum_{\alpha=1}^k \iota_{\psi'_\alpha}\eta_\theta^\alpha=\sum_{\alpha=1}^k\left(\parder{z^\alpha}{t^\alpha}-\rho^*\theta^\alpha(\phi'_\alpha)\right),
\]
which gives the second equation in \eqref{eq:split_HdDW_base}. 
\end{proof}

The main point in applications to second-order systems of PDEs is that, for particular types of Hamiltonian functions, the second equation in \eqref{eq:split_HdDW_base} becomes a balance equation for the dissipative variables, while the first one is a second-order system of PDEs for the coordinates of the projection $\phi$ of \(\psi\) onto \(P\). This projection may be independent of the dissipative variables,  or it may happen that dissipative variables can be fixed of a particular form so as to the balance equation of the dissipation variables always holds. The following result describes this situation for Hamiltonian functions that are affine in the dissipative variables.
\begin{corollary}\label{cor:affine_dissipative_variables}
In the setting of Proposition~\ref{prop:split_HdDW_contactification}, define \(h(x,z)=g(x)+\sum_{\alpha=1}^k a_\alpha(x)\,z^\alpha\), where \(x\in P\). Then,  equations  \eqref{eq:split_HdDW_base}  become
\begin{equation}\label{eq:split_HdDW_affine_base}
\sum_{\alpha=1}^k \iota_{\phi'_\alpha}\omega^\alpha=\left(\dd g+\sum_{\alpha=1}^k z^\alpha\,\dd a_\alpha+\sum_{\alpha=1}^k a_\alpha\,\rho^*\theta^\alpha\right)\circ\psi
\end{equation}
and
\begin{equation}\label{eq:split_HdDW_affine_diss}
\sum_{\alpha=1}^k \parder{z^\alpha}{t^\alpha}=\left(\sum_{\alpha=1}^k\rho^*\theta^\alpha(\phi'_\alpha)-g-\sum_{\alpha=1}^k a_\alpha z^\alpha\right)\circ\psi.
\end{equation} 
\end{corollary}

\begin{proof}
Indeed, \(\parder{h}{z^\alpha}=a_\alpha(x)\) and \(\dd_P h=\dd g+\sum_{\alpha=1}^k z^\alpha\,\dd a_\alpha\). Substituting these identities into Proposition~\ref{prop:split_HdDW_contactification} yields the result.
\end{proof}

\subsection{Adapted reduced two-contact phase spaces and partial regularity}
Let us explain a type of exact two-symplectic model that will appear in many applications in Section~\ref{sec:examples}. In those examples, the reduced exact two-symplectic phase space is of cotangent type, in the sense that one component of the primitive is the tautological one-form of a cotangent bundle, whereas the other one is the pull-back of a primitive on its base manifold. This structure is precisely what allows one to distinguish the field variables from the spatial momenta and, after imposing a suitable regularity condition on the Hamiltonian, to eliminate the latter and recover a second-order system of partial differential equations. We now formalise this observation.

\begin{definition}\label{def:adapted_exact_two_symplectic}
Let \(\tau:\T^*N\to N\) be the canonical projection of the cotangent bundle of an \(n\)-dimensional manifold $N$, let \(\theta_N\) be the Liouville one-form on \(\T^*N\), and let \(\theta_e\in\Omega^1(\T^*N)\) be basic (it is the pull-back to $N$ of a differential one-form on $N$). Then, \(\T^*N\) is endowed with {\it an adapted exact two-symplectic structure} if \(\bm\omega=-\dd\theta_e\otimes e_1-\dd\theta_N\otimes e_2\). The associated two-contactification is the manifold \(M:=\T^*N\times\mathbb R^2\) endowed with the two-contact form \(\bm\eta_\theta=(\dd z^t-\rho^*\theta_e)\otimes e_1+(\dd z^x-\rho^*\theta_N)\otimes e_2\), where \(\rho:M\to \T^*N\) is the canonical projection and \(z^t,z^x\) are the canonical coordinates on \(\mathbb R^2\), pulled back to \(M\) via the projection from $M$ onto \(\mathbb R^2\).
\end{definition}

Note that we assume that $\theta_e$ is a basic form on $\T^*N$, which restricts the class of adapted exact two-symplectic structures. This assumption is motivated by the fact that, in the applications in the following sections, $\theta_e$ relates the variables in  the base manifold $N$, while $\theta_N$ is able to distinguish the field variables from the spatial momenta.

In local coordinates \(y^1,\dots,y^n\) on \(N\), with induced fibre coordinates \(\pi_1,\dots,\pi_n\) on \(T^*N\), one has \(\theta_e=\sum_{i=1}^n (\theta_e)_i(y) \,\dd y^i\) and \(\theta_N=\sum_{i=1}^n \pi_i\,\dd y^i\), so that \(\bm\omega=\frac12\sum_{i,j=1}^n \Omega_{ij}(y)\,\dd y^i\wedge\dd y^j\otimes e_1+\sum_{i=1}^n \dd y^i\wedge \dd\pi_i\otimes e_2\), where \(\Omega_{ij}:=\partial_j(\theta_e)_i-\partial_i(\theta_e)_j\). We hereafter write in this section
$$
\bm\theta=\theta_1\otimes e_1+\theta_2\otimes e_2=\theta_e\otimes e_1+\theta_N\otimes e_2\in \Omega^1(\T^*N\times \mathbb{R}^2)
$$
and $\rho:\T^*N\times \mathbb{R}^2\rightarrow \T^*N$ for the canonical projection.

The roles of $\theta_e$ and $\theta_N$ are different: the component associated with \(\theta_N\) encodes the spatial momenta through the tautological cotangent structure, while the component associated with \(\theta_e\) encodes a pairing on the base manifold \(N\). This is the geometric mechanism underlying the examples on adapted two-contact phase spaces of cotangent type considered later on, and it is a key for describing their HdDW equations.

The next notion is the natural analogue of regularity for this reduced setting.

\begin{definition}\label{def:xregular_reduced}
Let \(M=\T^*N\times\mathbb R^2\) be the two-contactification of an adapted exact two-symplectic manifold on $\T^*N$, and let \(h\in C^\infty(M)\). We say that \(h\) is {\it \(\pi\)-regular} if the partial fibre derivative \(\mathbb F_\pi h:\T^*N\times\mathbb R^2\to \T N\times\mathbb R^2\), given by \((y^1,\dots,y^n,\pi_1,\dots,\pi_n,z^t,z^x)\mapsto \left(y^1,\dots,y^n,\parder{h}{\pi_1},\dots,\parder{h}{\pi_n},z^t,z^x\right)\), is a local diffeomorphism. If \(\mathbb F_\pi h\) is a global diffeomorphism, then \(h\) is said to be {\it \(\pi\)-hyperregular}.
\end{definition}

The following is a usual consequence of the inverse function theorem applied to the partial fibre derivative \(\mathbb F_\pi h\).

\begin{proposition}\label{prop:xregular_equivalent}
The following conditions are equivalent:
\begin{enumerate}
\item \(h\) is \(\pi\)-regular;
\item the Hessian matrix \(\left(\partial^2 h/\partial \pi_i\partial \pi_j\right)_{i,j=1,\dots,n}\) is non-singular at every point;
\item the relations \(v^i=\partial h/\partial \pi_i\), with \(i=1,\dots,n\), can be solved locally for the variables \(\pi_1,\dots,\pi_n\) as smooth functions of \((y^1,\dots,y^n,v^1,\dots,v^n,z^t,z^x)\).
\end{enumerate}
\end{proposition}

The relevance of $\pi$-regularity is that it allows one to reconstruct a second-order system of PDEs from the two-contact Hamilton–De Donder–Weyl equations on $M$.

\begin{proposition}\label{prop:reduced_second_order_equations}
Let \(M=\T^*N\times\mathbb R^2\) be the two-contactification of an adapted exact two-symplectic manifold, and let \(h\in C^\infty(M)\) be \(\pi\)-regular. Every solution, let us say \(\psi=(y^1,\dots,y^n,\pi_1,\dots,\pi_n,z^t,z^x):D\subset\mathbb R^2\to M\), of the two-contact Hamilton--De Donder--Weyl equations satisfies locally a second-order system of PDEs for the variables \(y^1,\dots,y^n\). More precisely, if the local inverse of the relations \((y^i)_x=\parder{h}{\pi_i}\) with \(i=1,\dots,n\) read \(\pi=\Pi(y,y_x,z^t,z^x)\), then
\begin{equation}\label{eq:reduced_second_order_system}
\sum_{j=1}^n \Omega_{ji}(y)\,(y^j)_t-\partial_x\!\bigl(\Pi_i(y,y_x,z^t,z^x)\bigr)=\left(\parder{h}{y^i}+\parder{h}{z^t}\,(\theta_e)_i(y)+\parder{h}{z^x}\,\pi_i\right)\Big|_{\pi=\Pi(y,y_x,z^t,z^x)},
\end{equation}
for  \(\Omega_{ki}:=\partial_i(\theta_e)_k-\partial_k(\theta_e)_i\) and $i,k=1,\dots,n$. Moreover, the dissipative variables satisfy
\begin{equation}\label{eq:reduced_dissipative_balance}
\partial_t z^t+\partial_x z^x=\sum_{i=1}^n (\theta_e)_i(y)\,(y^i)_t+\sum_{i=1}^n \Pi_i(y,y_x,z^t,z^x)\,(y^i)_x-h\bigl(y,\Pi(y,y_x,z^t,z^x),z^t,z^x\bigr).
\end{equation}
\end{proposition}

\begin{proof}
By Proposition~\ref{prop:split_HdDW_contactification}, the HdDW equations split into an equation on \(\T^*N\) and an equation for the dissipative variables. In the above coordinates, the first equation  in \eqref{eq:split_HdDW_base} reads
\[
\iota_{\phi'_t}\omega_1+\iota_{\phi'_x}\omega_2=\dd_{\T^*N}h+\parder{h}{z^t}\rho^*\theta_e+\parder{h}{z^x}\rho^*\theta_N.
\]
Since
\[
\omega_1=\frac12\sum_{i,j=1}^n \Omega_{ij}(y)\,\dd y^i\wedge\dd y^j,\qquad
\omega_2=\sum_{i=1}^n \dd y^i\wedge\dd\pi_i,
\]
one obtains
\begin{equation}\label{eq:relInver}
(y^i)_x=\parder{h}{\pi_i},\qquad i=1,\dots,n,
\end{equation}
and
\begin{equation}\label{eq:2}
\sum_{j=1}^n \Omega_{ji}(y)\,(y^j)_t-\partial_x\pi_i=\parder{h}{y^i}+\parder{h}{z^t}\,(\theta_e)_i(y)+\parder{h}{z^x}\,\pi_i,\qquad i=1,\dots,n.
\end{equation}
By \(\pi\)-regularity, the first relations \eqref{eq:relInver} can be solved locally as \(\pi_i=\Pi_i(y,y_x,z^t,z^x)\), for \(i=1,\dots,n\). Substituting them into  \eqref{eq:2} yields \eqref{eq:reduced_second_order_system}. Equation \eqref{eq:reduced_dissipative_balance} follows from Proposition~\ref{prop:split_HdDW_contactification}, since
\[
\rho^*\theta_1(\phi'_t)+\rho^*\theta_2(\phi'_x)=\sum_{i=1}^n (\theta_e)_i(y)\,(y^i)_t+\sum_{i=1}^n \pi_i (y^i)_x.
\]
\end{proof} 

The previous proposition explains the structure of the reduced examples studied later on. The component associated with \(\theta_{N}\) is responsible for the constitutive relations that identify the spatial momenta in terms of derivatives of solutions relative to independent variables, while the component associated with \(\theta_e\) determines the evolution pairing on the base manifold. The target system of second-order PDEs is then obtained by eliminating the momenta from the balance equations. In this way, adapted two-contactifications of exact two-symplectic manifolds provide a natural geometric mechanism to construct dissipative partial differential equations of second-order from two-contact Hamiltonian systems. In particular, the terms \(\parder{h}{z^x}\,\pi_i\) for $i=1,\ldots,n$ in \eqref{eq:reduced_second_order_system}  account for the additional reduced models in which the Hamiltonian depends on the spatial dissipative variable \(z^x\).

\begin{remark}\label{rem:affine_reduced_case}
If, in addition, \(h\) is affine in the dissipative variables, say
\[
h(y,\pi,z^t,z^x)=g(y,\pi)+a(y,\pi)z^t+b(y,\pi)z^x,
\]
then \eqref{eq:reduced_second_order_system} becomes
\begin{equation}\label{eq:reduced_second_order_system_affine}
\sum_{j=1}^n \Omega_{ji}(y)\,(y^j)_t-\partial_x\!\bigl(\Pi_i(y,y_x,z^t,z^x)\bigr)\!=\!\left(\parder{g}{y^i}+z^t\parder{a}{y^i}+z^x\parder{b}{y^i}+a\,(\theta_e)_i(y)+b\,\pi_i\right),
\end{equation}
for \(i=1,\dots,n\) and ${\pi=\Pi(y,y_x,z^t,z^x)}$, while \eqref{eq:reduced_dissipative_balance} becomes
\begin{multline*}
\partial_t z^t+\partial_x z^x=\sum_{i=1}^n (\theta_e)_i(y)\,(y^i)_t+\sum_{i=1}^n \Pi_i(y,y_x,z^t,z^x)\,(y^i)_x-g\bigl(y,\Pi(y,y_x,z^t,z^x)\bigr)\\-a\bigl(y,\Pi(y,y_x,z^t,z^x)\bigr)z^t-b\bigl(y,\Pi(y,y_x,z^t,z^x)\bigr)z^x.
\end{multline*}
 This expression covers both the examples with terms with a  dependence on \(z^t\) and those with a  dependence on \(z^x\) appearing in Section~\ref{sec:examples}.
\end{remark}

\subsection{Dissipation laws} \label{Sec:Dissipation}

In the $k$-contact setting, one is frequently interested in dissipative systems admitting quantities that are no longer conserved but still evolve according to a specific law \cite{CCM_18}. As a motivation, recall that if $X_h$ is a contact Hamiltonian vector field on $P$, then $\mathcal{L}_{X_h}h=-(Rh)h$ and more generally, a dissipative quantity is a function $F\in C^\infty(P)$ that satisfies that $\mathcal{L}_{X_h}F=-(Rh)F$. This shows that, in the contact case, the Hamiltonian evolves according to a  dissipation law governed by the Reeb dynamics.

Guided by this fact and by the corresponding notion of conservation law in field theory, one may introduce the concept of dissipation law for a $k$-contact Hamiltonian system $(M,\bm \eta ,H)$. For a map $F=(F^1,\dots,F^k)\colon M\to\mathbb{R}^k$, two related notions are defined. First, $F$ satisfies the \emph{dissipation law for sections} if, for every solution $\psi$ of the $k$-contact HdDW equations, the divergence of $F\circ\psi$ satisfies
\[
\operatorname{div}(F\circ\psi)=-\sum_{\alpha=1}^k\bigl[(R_\alpha h)F^\alpha\bigr]\circ\psi.
\]
Note that $\psi:U\subset \mathbb{R}^k\rightarrow M$ and $F\circ\psi:t\in U\subset \mathbb{R}^k\mapsto F\circ\psi(t)\in \mathbb{R}^k\simeq \T_t\mathbb{R}^k$ can be understood as a vector field, while $\mathbb{R}^k$ can be endowed with the standard volume form $\Omega=dt^1\wedge\ldots\wedge dt^k$ giving rise to the divergence of $F\circ \psi$. 

It is also worth analysing the meaning of the above expression. When the right-hand side is zero, this expression is what is called a {\it local conservation law for systems of PDEs}, where $F\circ \psi$ is called the {\it conserved current} \cite{OLV_86,Gun_87,OP2020}. Meanwhile, the right-hand side  measures how the left-hand side differs from a local conservation law, which is recovered when, in particular, the function $h$ is a first-integral of the Reeb vector fields.

Second, $F$ satisfies the \emph{dissipation law for $k$-vector fields} if, for every $k$-vector field $\mathbf{X}=(X_1,\ldots,X_k)$ solving the geometric HdDW equations for $k$-vector fields, one has
\[
\sum_{\alpha=1}^k\mathcal{L}_{X_\alpha}F^\alpha=-\sum_{\alpha=1}^k(R_\alpha h)F^\alpha.
\]
It can be proved that  if $F$ satisfies the dissipation law for sections, then the above identity for $k$-vector fields holds for every integrable section of $\mathbf{X}$; conversely, if the identity holds for a given integrable $k$-vector field, then its integral sections satisfy the corresponding divergence law.  

The following results are relevant as they explain how infinitesimal dynamical symmetries generate dissipation laws (see \cite[Proposition 6.6 \& Proposition 6.10]{GGMRR_20}). 

\begin{definition} An infinitesimal dynamical symmetry of a Hamiltonian $k$-contact manifold $(M,\bm\eta,h)$ is a vector field $Y$ on $M$ such that
$$
\mathcal{L}_Y\bm\eta=0,\qquad \mathcal{L}_Yh=0.
$$
\end{definition}

\begin{theorem} If $Y$ is an infinitesimal dynamical symmetry, then
\(
{\bm F}=-\iota_Y{\bm \eta}
\) 
satisfies the dissipation law for $k$-vector fields.
\end{theorem}

Hence, infinitesimal dynamical symmetries generate canonical dissipated quantities in complete analogy with the rôle played by symmetries and conservation laws in conservative theories, but now adapted to the genuinely dissipative $k$-contact framework.

\section{Applications}
\label{sec:examples}

The aim of this section is to show that the $k$-contact Hamilton--De Donder--Weyl formalism provides a natural framework for physically relevant field theories with dissipation and attenuation. In contrast with conservative multisymplectic or $k$-symplectic approaches, dissipative variables make it possible to encode nonconservative contributions directly at the geometric level \cite{GGMRR_20,Riv_22}. This makes the formalism especially suitable for models such as damped wave, Klein--Gordon-, or telegrapher-type equations, as well as for more general media in which dissipation is not imposed externally but evolves as part of the dynamics itself. The examples considered below should therefore be understood not only as concrete applications, but also as evidence that $k$-contact geometry offers a promising geometric language for a broader class of nonconservative field theories.
\subsection{A damped Klein--Gordon model}
\label{subsec:dampedwave}
Consider the Hamiltonian function
\(
h\in C^\infty(\mathcal{J}_{\mathbb{R},2})
\)
given by
\[
h(u,p^{t},p^{x},z^{t},z^{x})
=
\frac12\Bigl((p^{t})^{2}-\frac{1}{c^{2}}(p^{x})^{2}\Bigr)
+\frac12(z^{t})^{2}
+\frac{\epsilon}{2}u^{2},
\qquad c>0,\ \epsilon\geq 0,
\]
where \((u,p^{t},p^{x},z^{t},z^{x})\) are canonical Darboux coordinates on \((\mathcal{J}_{\mathbb{R},2},\bm\eta_{\mathbb{R},2},\mathcal{V}_{\mathbb{R},2})\).

Hence, the two-contact Hamilton--De Donder--Weyl equations for a section \(\psi=\psi(t,x)\) read
\[
u_{t}=\frac{\partial h}{\partial p^{t}},\qquad
u_{x}=\frac{\partial h}{\partial p^{x}},\qquad
\partial_{t}p^{t}+\partial_{x}p^{x}
=
-\Bigl(\frac{\partial h}{\partial u}
+p^{t}\frac{\partial h}{\partial z^{t}}
+p^{x}\frac{\partial h}{\partial z^{x}}\Bigr),
\]
together with
\[
\partial_{t}z^{t}+\partial_{x}z^{x}
=
p^{t}\frac{\partial h}{\partial p^{t}}
+p^{x}\frac{\partial h}{\partial p^{x}}
-h.
\]
Since $h$ is regular, one can write the momenta in terms of $u_t,u_x$. Indeed, \(\partial h/\partial p^{t}=p^{t}\) and \(\partial h/\partial p^{x}=-p^{x}/c^{2}\), and the first two equations give \(p^{t}=u_{t}\) and \(p^{x}=-c^{2}u_{x}\). On the other hand, \(\partial h/\partial u=\epsilon u\), \(\partial h/\partial z^{t}=z^{t}\), and \(\partial h/\partial z^{x}=0\), so the momentum balance equation becomes
\[
\partial_{t}p^{t}+\partial_{x}p^{x}=-(\epsilon u+z^{t}p^{t}).
\]
Substituting the above expressions for the momenta, one obtains
\[
u_{tt}-c^{2}u_{xx}+z^{t}u_{t}+\epsilon u=0.
\]
Hence, \(h\) yields a damped Klein--Gordon-type equation with effective damping coefficient \(z^{t}\) and restoring term \(\epsilon u\). The equation gives also a telegrapher-type equation describing signal propagation in transmission lines with resistance and inductance \cite{SakataWakasugi2017,deLucasLangeSardonRivas2026}. Since the Hessian of $h$ relative to momenta has two diagonal values of opposite sign, the resulting PDE is hyperbolic.

The evolution of the dissipative variables is governed by
\[
\partial_{t}z^{t}+\partial_{x}z^{x}
=
p^{t}\frac{\partial h}{\partial p^{t}}
+p^{x}\frac{\partial h}{\partial p^{x}}
-h
=
\frac12(p^{t})^{2}
-\frac{1}{2c^{2}}(p^{x})^{2}
-\frac12(z^{t})^{2}
-\frac{\epsilon}{2}u^{2},
\]
that is,
\[
\partial_{t}z^{t}+\partial_{x}z^{x}
=
\frac12u_{t}^{2}
-\frac12c^{2}u_{x}^{2}
-\frac12(z^{t})^{2}
-\frac{\epsilon}{2}u^{2}.
\]
Therefore, $z^t$   evolves according to a balance law. Anyhow, one may assume that $z^t(t,x)=\lambda(t,x)$ is an arbitrary function and $z^x(t,x)$ can be determined from the balance for the dissipation variables and each solution $u(t,x)$ of the damped Klein--Gordon equation
$$
u_{tt}-c^2u_{xx}+\lambda(t,x)u_t+\epsilon u=0.
$$

This example may be regarded as a natural two-contact candidate for damped Klein--Gordon or telegrapher equations \cite{Wakasugi2012,AncoMarquezGarridoGandarias2024}.

Let us now interpret dissipation in the sense of $k$-contact geometry. By the general notion of dissipation law for $k$-contact Hamiltonian systems \cite{GGMRR_20}, a map \(F=(F^{t},F^{x})\colon M\to\mathbb{R}^{2}\) is said to be dissipated along solutions if
\[
\partial_{t}(F^{t}\circ\psi)+\partial_{x}(F^{x}\circ\psi)
=
-\bigl[(\Lie_{R_{t}}h)F^{t}+(\Lie_{R_{x}}h)F^{x}\bigr]\circ\psi.
\]
In the present case, \(R_{t}=\partial/\partial z^{t}\) and \(R_{x}=\partial/\partial z^{x}\), so
\(
\Lie_{R_{t}}h=z^{t},  \Lie_{R_{x}}h=0,
\)
and consequently
\[
\partial_{t}(F^{t}\circ\psi)+\partial_{x}(F^{x}\circ\psi)
=
-(z^{t}F^{t})\circ\psi.
\]
Thus, \(z^{t}\) plays the role of an instantaneous dissipation rate.

When \(\epsilon=0\), the Hamiltonian is independent of \(u\), and therefore \(Y=\partial/\partial u\) is an infinitesimal dynamical symmetry. The associated dissipated quantity is \(F=(-\iota_Y\eta^t,-\iota_Y\eta^x)=(p^t,p^x)\). Hence, the general dissipation theorem yields
\[
\partial_t p^t+\partial_x p^x=-z^t p^t.
\]
Using \(p^t=u_t\) and \(p^x=-c^2u_x\), one obtains
\[
\partial_t(u_t)+\partial_x(-c^2u_x)=-z^t u_t.
\]
Thus, for \(z^t=\lambda(t,x)\), one recovers the standard local momentum-type balance law for the damped wave equation \(u_{tt}-c^2u_{xx}+\lambda(t,x)u_t=0\). In the particular case of constant damping \(z^t=\lambda\), and under suitable boundary conditions, this yields
\[
\frac{d}{dt}\int p^t\,dx=-\lambda\int p^t\,dx,
\qquad
e^{\lambda t}\int p^t\,dx=\mathrm{const}.
\]
Therefore, the \(k\)-contact formalism geometrises a classical exponentially weighted balance law for damped wave equations. This is consistent with the generalized momentum-type conservation laws described in the symmetry analysis of one-dimensional damped wave equations, and also with the conformal balance laws arising in damped acoustic wave models written in multi-symplectic form \cite{CaiZhangWang2017}.
\subsection{Reaction--diffusion / Allen--Cahn type equation}
\label{subsec:allencahn}

Consider the equation
\begin{equation}\label{eq:RD_AC}
u_t-\nu u_{xx}+V'(u)+\lambda u=0,\qquad \nu>0,\ \lambda\ge0,
\end{equation}
where \(V\in C^\infty(\mathbb R)\), for instance \(V(u)=\frac14(u^2-1)^2\). This class includes, in particular, Allen--Cahn type reaction--diffusion equations \cite{AllenCahn1979,Murray2002,Fife1979}.

Let us study \eqref{eq:RD_AC} via the two-contact manifold obtained as the contactification of a reduced exact two-symplectic manifold. Let \(P_0=\mathbb R^4\) with coordinates \((u,v,p^x,q^x)\), endowed with the one-form 
\[\bm\theta=\frac12(-v\,\dd u+u\,\dd v)\otimes e_1+(p^x\,\dd u+q^x\,\dd v)\otimes e_2\]

Hence, one obtains an exact two-symplectic manifold $(P_0,\bm\omega=-\dd\bm\theta)$. 
 Let us now consider the contactification \(M=P_0\times\mathbb R^2\) with coordinates \((u,v,p^x,q^x,z^t,z^x)\) and the two-contact form \({\bm \eta}=(\dd z^t-\theta^t)\otimes e_1+(\dd z^x-\theta^x)\otimes e_2.\) Then, consider  the Hamiltonian
\[
h_{\mathrm{RD}}=-\frac1{\nu}\,p^xq^x+v\bigl(V'(u)+\lambda u\bigr).
\]
Since \(h_{\mathrm{RD}}\) is independent of \(z^t\) and \(z^x\), one has \( {R_t}h_{\mathrm{RD}}= {R_x}h_{\mathrm{RD}}=0\). Hence, the two-contact Hamilton--De Donder--Weyl equations reduce to
\[
\partial_t v-\partial_x p^x=\frac{\partial h_{\mathrm{RD}}}{\partial u},\qquad
-\partial_t u-\partial_x q^x=\frac{\partial h_{\mathrm{RD}}}{\partial v},\qquad
\partial_x u=\frac{\partial h_{\mathrm{RD}}}{\partial p^x},\qquad
\partial_x v=\frac{\partial h_{\mathrm{RD}}}{\partial q^x}.
\]
A direct computation gives
\[
\frac{\partial h_{\mathrm{RD}}}{\partial u}=v\bigl(V''(u)+\lambda\bigr),\qquad
\frac{\partial h_{\mathrm{RD}}}{\partial v}=V'(u)+\lambda u,
\]
\[
\frac{\partial h_{\mathrm{RD}}}{\partial p^x}=-\frac1{\nu}q^x,\qquad
\frac{\partial h_{\mathrm{RD}}}{\partial q^x}=-\frac1{\nu}p^x.
\]
Therefore, \(q^x=-\nu u_x\) and \(p^x=-\nu v_x\). Substituting these expressions into the previous equations, one obtains
\[
u_t-\nu u_{xx}+V'(u)+\lambda u=0,\qquad
v_t+\nu v_{xx}=v\bigl(V''(u)+\lambda\bigr).
\]
Hence the first equation is precisely \eqref{eq:RD_AC}, while the second one describes the complementary evolution of the auxiliary field \(v\).

In particular, \(v=0\) is a consistent reduction, and every solution \(u\) of \eqref{eq:RD_AC} defines a solution of the full Hamiltonian system by taking \(v=0\), \(q^x=-\nu u_x\), and \(p^x=0\).

Thus, \eqref{eq:RD_AC} admits a natural two-contact Hamilton--De Donder--Weyl description on a reduced phase space. Since \(h_{\mathrm{RD}}\) is independent of the dissipative variables, there is no explicit dissipation through the Reeb directions in this model.  
\subsubsection{A family of generalized Burgers-type equations in a two-contact setting}

We now show that the previous two-contact construction can be modified so as to cover a wider class of Burgers-type equations. Consider the reduced phase space \(P_0\) with coordinates \((u,v,p_u^x,p_v^x)\), endowed with the one-form with values in $\mathbb{R}^2$ given by
$$
\bm \theta=\theta_e\otimes e_1+\theta_\mathbb{R}\otimes e_2,
$$
where
\(\theta_e=\frac12(-v\,\dd u+u\,\dd v)\) and \(\theta_\mathbb{R}=p_u^x\,\dd u+p_v^x\,\dd v\). Then, 
\[
\bm\omega=-\dd \bm\theta=-\dd\theta_e\otimes e_1-\dd\theta_N \otimes e_2=-\dd u\wedge \dd v\otimes e_1+(\dd u\wedge \dd p_u^x+\dd v\wedge \dd p_v^x)\otimes e_2.
\]
We consider the two-contactification \(M=P_0\times\mathbb R^2\) with coordinates \((u,v,p_u^x,p_v^x,z^t,z^x)\) and canonical two-contact forms \(\bm\eta= (\dd z^t-\rho^*\theta_e)\otimes e_1+(\dd z^x-\rho^*\theta_N)\otimes e_2\), where $\rho:M\rightarrow P_0$ is the standard projection.

Let \(D,B,C\in C^\infty(\mathbb R)\), with \(D(u)\neq0\), and consider the Hamiltonian
\[
h=
-\frac{1}{D(u)}\,p_u^x p_v^x-\frac{B(u)}{D(u)}\,z^x+v\,C(u).
\]
The associated two-contact Hamilton--De Donder--Weyl equations read
\[
\partial_t v-\partial_x p_u^x=\frac{\partial h}{\partial u}+p_u^x\frac{\partial h}{\partial z^x},\qquad
-\partial_t u-\partial_x p_v^x=\frac{\partial h}{\partial v}+p_v^x\frac{\partial h}{\partial z^x},
\]
together with \(\partial_x u=\frac{\partial h}{\partial p_u^x}\) and \(\partial_x v=\frac{\partial h}{\partial p_v^x}\). Since
\[
\frac{\partial h}{\partial v}=C(u),\qquad
\frac{\partial h}{\partial z^x}=-\frac{B(u)}{D(u)},\qquad
\frac{\partial h}{\partial p_u^x}=-\frac{1}{D(u)}\,p_v^x,\qquad
\frac{\partial h}{\partial p_v^x}=-\frac{1}{D(u)}\,p_u^x,
\]
the last two equations yield \(p_v^x=-D(u)\,u_x\) and \(p_u^x=-D(u)\,v_x\). Moreover, the equation for \(u\) becomes
\[
-\partial_t u-\partial_x p_v^x=C(u)-\frac{B(u)}{D(u)}\,p_v^x.
\]
Substituting \(p_v^x=-D(u)\,u_x\), we obtain
\begin{equation}\label{eq:generalized_burgers_family}
u_t-\partial_x\bigl(D(u)\,u_x\bigr)+B(u)\,u_x+C(u)=0.
\end{equation}
Hence, this two-contact Hamiltonian system induces a whole family of generalized Burgers-type equations of convection--diffusion--reaction type \cite{PopovicPtashnykSattar2025}.

Several relevant cases are recovered immediately. If \(D(u)=\nu>0\) is constant and \(C(u)=0\), then \eqref{eq:generalized_burgers_family} reduces to \(u_t+B(u)\,u_x=\nu u_{xx}\). In particular, for \(B(u)=u\) one recovers the viscous Burgers equation \(u_t+u\,u_x=\nu u_{xx}\) (see for instance \cite{Salih2016BurgersEquation}). More generally, if \(D(u)=\nu\) is constant, then \eqref{eq:generalized_burgers_family} becomes \(u_t+B(u)\,u_x=\nu u_{xx}-C(u)\), which includes Burgers-type equations with reaction terms, namely with $C(u)\neq 0$ (cf. \cite{KayaElSayed2004}).

This shows that the two-contact approach is flexible enough to accommodate a nontrivial family of nonlinear scalar equations while keeping the underlying geometric structure fixed and modifying only the Hamiltonian.
 
\subsection{Porous medium equation with linear absorption}
\label{subsec:PME}

Consider the porous medium equation with linear absorption
\begin{equation}\label{eq:PME}
u_t=(u^m)_{xx}-bu,\qquad m>1,\qquad b\geq 0,
\end{equation}
which is a standard model of nonlinear diffusion in porous media with absorption; see, for instance, \cite{Barenblatt1952,Aronson1986,Vazquez2007,BandleNanbuStakgold1998,Kwak1998}. For simplicity, we restrict ourselves to $u>0$.

Let us describe \eqref{eq:PME} on a two-contactification. Consider \(P_0=\mathbb R^4\simeq \T^*\mathbb{R}^2\) with global coordinates \((u,v,p^x,q^x)\), and define
\[
\bm\theta=\frac12(-v\,\dd u+u\,\dd v)\otimes e_1+(p^x\,\dd u+q^x\,\dd v)\otimes e_2.
\]
This gives rise to an exact two-symplectic manifold \((P_0,\bm\omega=-\dd\bm\theta)\) with
\[
\bm\omega=(-\dd u\wedge\dd v)\otimes e_1+(\dd u\wedge\dd p^x+\dd v\wedge\dd q^x)\otimes e_2.
\]
Consider the contactification \(M=P_0\times\mathbb R^2\) with coordinates \((u,v,p^x,q^x,z^t,z^x)\) and two-contact form
\[
\bm\eta=(\dd z^t-\rho^*\theta^1)\otimes e_1+(\dd z^x-\rho^*\theta^2)\otimes e_2,
\]
where \(\rho:M\to P_0\) is the canonical projection. Define the Hamiltonian \(h_{\mathrm{PME}}\in C^\infty(P_0\times\mathbb R^2)\) by
\[
h_{\mathrm{PME}}=-\frac{1}{m u^{m-1}}\,p^x q^x+2b z^t.
\]
Then, \(R_th_{\mathrm{PME}}=2b\) and \(R_xh_{\mathrm{PME}}=0\). Hence, the associated Hamilton--De Donder--Weyl equations are
\[
\partial_t v-\partial_x p^x=\frac{\partial h_{\mathrm{PME}}}{\partial u}- b   v,\quad
-\partial_t u-\partial_x q^x=\frac{\partial h_{\mathrm{PME}}}{\partial v}+ b  u,\quad
\partial_x u=\frac{\partial h_{\mathrm{PME}}}{\partial p^x},\quad
\partial_x v=\frac{\partial h_{\mathrm{PME}}}{\partial q^x}.
\]
A direct computation gives
\[
\frac{\partial h_{\mathrm{PME}}}{\partial u}=\frac{m-1}{m u^m}\,p^x q^x,\quad
\frac{\partial h_{\mathrm{PME}}}{\partial v}=0,\quad
\frac{\partial h_{\mathrm{PME}}}{\partial p^x}=-\frac{1}{m u^{m-1}}\,q^x,\quad
\frac{\partial h_{\mathrm{PME}}}{\partial q^x}=-\frac{1}{m u^{m-1}}\,p^x.
\]
Therefore, \(q^x=-m u^{m-1}u_x=-(u^m)_x\), while the second balance equation becomes \(-u_t-\partial_x q^x=b u\). Substituting the expression for \(q^x\), one recovers precisely \eqref{eq:PME}. Hence, \eqref{eq:PME} admits a natural two-contact Hamilton--De Donder--Weyl description on the region where \(u\neq0\) since $m>1$.

The first equation determines the complementary evolution of the auxiliary field \(v\), namely \(\partial_t v-\partial_x p^x=\frac{m-1}{m u^m}\,p^x q^x-b v\). In particular, \(v=0\) and \(p^x=0\) define a consistent reduction, and every solution \(u\) of \eqref{eq:PME} lifts to a solution of the full Hamiltonian system by taking \(q^x=-(u^m)_x\), \(v=0\), and \(p^x=0\).

Moreover, the dissipative variables satisfy
\[
\partial_t z^t+\partial_x z^x=\theta^t(\phi'_t)+\theta^x(\phi'_x)-h_{\mathrm{PME}}
=\frac12(-v u_t+u v_t)+p^x u_x+q^x v_x-h_{\mathrm{PME}}.
\]
Using the constitutive relations for \(u_x\) and \(v_x\), this becomes
\[
\partial_t z^t+\partial_x z^x=\frac12(-v u_t+u v_t)-\frac{1}{m u^{m-1}}\,p^x q^x-2bz^t.
\]

Under the consistent reduction \(v=0\) and \(p^x=0\), this becomes \(\partial_t z^t+\partial_x z^x=-2b z^t\). Thus the linear absorption term is encoded through the time-like contact variable \(z^t\), while the singularity of \(h_{\mathrm{PME}}\) at \(u=0\) reflects the well-known degeneracy of the porous medium equation at vanishing density.


\subsection{Complex Ginzburg--Landau equation} \label{subsec:CGL} Consider the complex Ginzburg--Landau equation \begin{equation}\label{eq:CGL} \psi_t=(1+i\alpha)\psi_{xx}+(1-\gamma)\psi-(1+i\beta)|\psi|^2\psi,\qquad \alpha,\beta,\gamma\in\mathbb R, \end{equation} which is a standard model for nonlinear waves and pattern formation in dissipative media \cite{KuramotoTsuzuki1976,AransonKramer2002,CrossHohenberg1993}, and its variants with an additional linear gain/loss term also arise naturally in dissipative nonlinear optics and mode-locked laser theory \cite{AkhmedievSotoCrespoTown2001,KomarovLeblondSanchez2005}. Writing \(\psi=a+ib\), one obtains \[ a_t=a_{xx}-\alpha b_{xx}+R_a(a,b)-\gamma a,\qquad b_t=\alpha a_{xx}+b_{xx}+R_b(a,b)-\gamma b, \] where \(R_a(a,b)=a-(a^2+b^2)a+\beta(a^2+b^2)b\) and \(R_b(a,b)=b-\beta(a^2+b^2)a-(a^2+b^2)b\). It is convenient to write the diffusion part in matrix form as \[ \binom{a_t}{b_t} = D\binom{a_{xx}}{b_{xx}}+\binom{R_a(a,b)-\gamma a}{R_b(a,b)-\gamma b}, \qquad D= \begin{pmatrix} 1&-\alpha\\ \alpha&1 \end{pmatrix}, \] so that \[ D^{-1}=\frac{1}{1+\alpha^2} \begin{pmatrix} 1&\alpha\\ -\alpha&1 \end{pmatrix}. \] Instead of introducing spatial derivatives as independent constrained variables, we consider a reduced phase space \(P_0\) with coordinates \((a,b,c,d,p_a^x,p_b^x,q_a^x,q_b^x)\), where \(c,d\) are auxiliary fields. On \(P_0\) define \[ \bm \theta=\frac12(-c\,\dd a+a\,\dd c-d\,\dd b+b\,\dd d)\otimes e_1+(p_a^x\,\dd a+p_b^x\,\dd b+q_a^x\,\dd c+q_b^x\,\dd d)\otimes e_2. \] Then, one has that
\[\bm\omega=-d\bm\theta=(-\dd a\wedge \dd c-\dd b\wedge \dd d)\otimes e_1+(\dd a\wedge \dd p_a^x+\dd b\wedge \dd p_b^x+\dd c\wedge \dd q_a^x+\dd d\wedge \dd q_b^x)\otimes e_2.
\]
Consider the two-contactification \(M=P_0\times\mathbb R^2\) with coordinates \((a,b,c,d,p_a^x,p_b^x,q_a^x,q_b^x,z^t,z^x)\) and canonical two-contact form given by  \(\bm\eta=(\dd z^t-\rho^*\theta_e)\otimes e_1+(\dd z^x-\rho^*\theta_{P_0})\otimes e_2\), where $\rho:M\rightarrow P_0$ is the standard projection. We take the Hamiltonian \[ h_{\mathrm{CGL}} = -\frac{1}{1+\alpha^2}\Bigl(p_a^x(q_a^x+\alpha q_b^x)+p_b^x(-\alpha q_a^x+q_b^x)\Bigr) -c\,R_a(a,b)-d\,R_b(a,b)+2\gamma z^t. \] Since \(\partial h_{\mathrm{CGL}}/\partial z^t=2\gamma\) and \(\partial h_{\mathrm{CGL}}/\partial z^x=0\), the two-contact Hamilton--De Donder--Weyl equations become \[ \partial_t c-\partial_x p_a^x=\frac{\partial h_{\mathrm{CGL}}}{\partial a}-\gamma c,\qquad -\partial_t a-\partial_x q_a^x=\frac{\partial h_{\mathrm{CGL}}}{\partial c}+\gamma a, \] \[ \partial_t d-\partial_x p_b^x=\frac{\partial h_{\mathrm{CGL}}}{\partial b}-\gamma d,\qquad -\partial_t b-\partial_x q_b^x=\frac{\partial h_{\mathrm{CGL}}}{\partial d}+\gamma b, \] together with \[ \partial_x a=\frac{\partial h_{\mathrm{CGL}}}{\partial p_a^x},\qquad \partial_x b=\frac{\partial h_{\mathrm{CGL}}}{\partial p_b^x}. \] A direct computation gives \(\partial h_{\mathrm{CGL}}/\partial c=-R_a(a,b)\), \(\partial h_{\mathrm{CGL}}/\partial d=-R_b(a,b)\), and \[ \partial_x a=-\frac{1}{1+\alpha^2}(q_a^x+\alpha q_b^x),\qquad \partial_x b=-\frac{1}{1+\alpha^2}(-\alpha q_a^x+q_b^x). \] Solving for \(q_a^x\) and \(q_b^x\), one obtains \[ q_a^x=-a_x+\alpha b_x,\qquad q_b^x=-\alpha a_x-b_x. \] Substituting into the equations for \(a\) and \(b\), we get \[ a_t=-\partial_x q_a^x+R_a(a,b)-\gamma a,\qquad b_t=-\partial_x q_b^x+R_b(a,b)-\gamma b, \] which reproduces precisely the real form of \eqref{eq:CGL}. Therefore, this complex Ginzburg--Landau equation with linear gain/loss admits a natural two-contact Hamilton--De Donder--Weyl description on a reduced phase space, now with a Hamiltonian affine in the dissipative variable \(z^t\). The remaining equations describe the evolution of the auxiliary fields \(c\) and \(d\). In particular, \(c=d=p_a^x=p_b^x=0\) is a consistent reduction, and every solution \((a,b)\) of \eqref{eq:CGL} determines a solution of the full Hamiltonian system by taking \(q_a^x=-a_x+\alpha b_x\), \(q_b^x=-\alpha a_x-b_x\).


\subsection{Damped nonlinear Schr\"odinger equation in real form} \label{subsec:dNLS} Consider the damped nonlinear Schr\"odinger equation \begin{equation}\label{eq:damped_NLS} i\psi_t+\psi_{xx}+|\psi|^2\psi+i\gamma\psi=0,\qquad \gamma\ge0, \end{equation} which may be regarded as a dissipative variant of the nonlinear Schr\"odinger equation, a standard model in nonlinear optics, while damped and driven versions play a central role in dissipative cavity optics \cite{SulemSulem1999,LugiatoLefever1987,KippenbergEtAl2018}. Writing \(\psi=a+ib\), one obtains \[ a_t=-b_{xx}-(a^2+b^2)b-\gamma a,\qquad b_t=a_{xx}+(a^2+b^2)a-\gamma b. \] Set \(R_a(a,b)=-(a^2+b^2)b-\gamma a\) and \(R_b(a,b)=(a^2+b^2)a-\gamma b\). Then the system reads \(a_t=-b_{xx}+R_a(a,b)\), \(b_t=a_{xx}+R_b(a,b)\). Instead of introducing derivative constraints by hand, we consider a reduced phase space \(P_0\) with coordinates \((a,b,c,d,p_a^x,p_b^x,q_a^x,q_b^x)\), endowed with the one-form taking values in $\mathbb{R}^2$ of the form
\[ 
\bm\theta=\frac12(-c\,\dd a+a\,\dd c-d\,\dd b+b\,\dd d)\otimes e_1+(p_a^x\,\dd a+p_b^x\,\dd b+q_a^x\,\dd c+q_b^x\,\dd d)\otimes e_2. 
\] 
Then, 
\[\bm \omega=-(\dd a\wedge \dd c+\dd b\wedge \dd d)\otimes e_1+(\dd a\wedge \dd p_a^x+\dd b\wedge \dd p_b^x+\dd c\wedge \dd q_a^x+\dd d\wedge \dd q_b^x)\otimes e_2.
\]
Define the two-contactification \(M=P_0\times\mathbb R^2\) with coordinates \((a,b,c,d,p_a^x,p_b^x,q_a^x,q_b^x,z^t,z^x)\) and  two-contact form  \(\bm\eta=(\dd z^t-\rho^*\theta_e)\otimes e_1+(\dd z^x-\rho^*\theta^x)\otimes e_2,\) where $\rho:M
\rightarrow P_0$ is the canonical projection. 
We take the Hamiltonian 
\[ h_{\mathrm{dNLS}}=-p_a^x q_b^x+p_b^x q_a^x-c\,R_a(a,b)-d\,R_b(a,b). \] 
Since \(h_{\mathrm{dNLS}}\) is independent of \(z^t\) and \(z^x\), one has \( R_t h_{\mathrm{dNLS}}= R_x h_{\mathrm{dNLS}}=0\). Hence, the two-contact Hamilton--De Donder--Weyl equations reduce to 
\[ \partial_t c-\partial_x p_a^x=\frac{\partial h_{\mathrm{dNLS}}}{\partial a},\qquad -\partial_t a-\partial_x q_a^x=\frac{\partial h_{\mathrm{dNLS}}}{\partial c}, \] 
\[ \partial_t d-\partial_x p_b^x=\frac{\partial h_{\mathrm{dNLS}}}{\partial b},\qquad -\partial_t b-\partial_x q_b^x=\frac{\partial h_{\mathrm{dNLS}}}{\partial d}, \] 
together with 
\[ \partial_x a=\frac{\partial h_{\mathrm{dNLS}}}{\partial p_a^x}=-q_b^x,\qquad \partial_x b=\frac{\partial h_{\mathrm{dNLS}}}{\partial p_b^x}=q_a^x. \] 
Therefore, \(q_a^x=b_x\) and \(q_b^x=-a_x\). Moreover, \(\partial h_{\mathrm{dNLS}}/\partial c=-R_a(a,b)\) and \(\partial h_{\mathrm{dNLS}}/\partial d=-R_b(a,b)\), so the equations for \(a\) and \(b\) become \[ a_t=-\partial_x q_a^x+R_a(a,b),\qquad b_t=-\partial_x q_b^x+R_b(a,b). \] Substituting the expressions for \(q_a^x\) and \(q_b^x\), one recovers precisely the real form of \eqref{eq:damped_NLS}. Hence the damped nonlinear Schr\"odinger equation admits a natural \(2\)-contact Hamilton--De Donder--Weyl description on a reduced phase space. The remaining equations determine the auxiliary fields \(c\) and \(d\). In particular, \(c=d=p_a^x=p_b^x=0\) is a consistent reduction, and every solution \((a,b)\) of \eqref{eq:damped_NLS} lifts to a solution of the full Hamiltonian system by taking \(q_a^x=b_x\), \(q_b^x=-a_x\), with \(z^t\) and \(z^x\) constant. The damping parameter \(\gamma\) is already incorporated into the reaction terms \(R_a\) and \(R_b\). Therefore, the present construction does not use the contact variables to generate dissipation through the Reeb directions. If one wishes to introduce genuinely two-contact dissipation laws, one should allow the Hamiltonian to depend explicitly on \(z^t\) or \(z^x\), which would produce a different nonconservative deformation of \eqref{eq:damped_NLS}.
 
\subsection{Fisher--KPP equation with linear loss}
\label{subsec:Fisher}

Consider the reaction--diffusion equation
\begin{equation}\label{eq:FisherKPP_loss}
u_t=\mu u_{xx}+ru(1-u)-\lambda u,\qquad \mu>0,\ r\ge0,\ \lambda\ge0,
\end{equation}
which may be regarded as a dissipative variant of the classical Fisher--KPP equation \cite{Fisher1937,KurataShi2008,Wang2016,OrugantiShiShivaji2002,RoquesChekroun2007}. This equation describes for particular cases of the coefficients several natural processes like the expansion of advantageous genes \cite{Fisher1937}.

Let us describe \eqref{eq:FisherKPP_loss} on the $k$-symplectic part of its $k$-contactification. Let \(P_0\) be the manifold with coordinates \((u,v,p^x,q^x)\), endowed with the one-form \(\bm \theta=\frac12(-v\,\dd u+u\,\dd v)\otimes e_1+(p^x\,\dd u+q^x\,\dd v)\otimes e_2\). Then, \(\bm\omega=-\dd\bm \theta=-\dd u\wedge \dd v\otimes e_1+(\dd u\wedge \dd p^x+\dd v\wedge \dd q^x)\otimes e_2\). We consider the two-contactification \(M=P_0\times\mathbb R^2\) with coordinates \((u,v,p^x,q^x,z^t,z^x)\) and  the canonical two-contact form \(\bm\eta=(\dd z^t-\rho^*\theta^1)\otimes e_1+(\dd z^x-\rho^*\theta^2)\otimes e_2\), where $\rho:P_0\times \mathbb{R}^2\rightarrow P_0$ is the natural projection.

We take the Hamiltonian
\[
h_{\mathrm{FKPP}}=-\frac1{\mu}\,p^xq^x+v\bigl(-ru(1-u)+\lambda u\bigr).
\]
Since \(h_{\mathrm{FKPP}}\) is independent of \(z^t\) and \(z^x\), one has \( {R_t}h_{\mathrm{FKPP}}= {R_x}h_{\mathrm{FKPP}}=0\). Hence, the two-contact Hamilton--De Donder--Weyl equations reduce to
\[
\partial_t v-\partial_x p^x=\frac{\partial h_{\mathrm{FKPP}}}{\partial u},\qquad
-\partial_t u-\partial_x q^x=\frac{\partial h_{\mathrm{FKPP}}}{\partial v},\qquad
\partial_x u=\frac{\partial h_{\mathrm{FKPP}}}{\partial p^x},\qquad
\partial_x v=\frac{\partial h_{\mathrm{FKPP}}}{\partial q^x}.
\]
A direct computation gives \(\partial h_{\mathrm{FKPP}}/\partial v=-ru(1-u)+\lambda u\), \(\partial h_{\mathrm{FKPP}}/\partial p^x=-q^x/\mu\), and \(\partial h_{\mathrm{FKPP}}/\partial q^x=-p^x/\mu\). Therefore \(q^x=-\mu u_x\) and \(p^x=-\mu v_x\), while the equation for \(u\) becomes
\[
-\partial_t u-\partial_x q^x=-ru(1-u)+\lambda u.
\]
Substituting \(q^x=-\mu u_x\), one recovers precisely \eqref{eq:FisherKPP_loss}. Hence, the Fisher--KPP equation with linear loss admits a natural two-contact Hamilton--De Donder--Weyl description on a reduced phase space.

The remaining equation determines the evolution of the auxiliary field \(v\). In particular, \(v=0\) and \(p^x=0\) define a consistent reduction, and every solution \(u\) of \eqref{eq:FisherKPP_loss} lifts to a solution of the full Hamiltonian system by taking \(q^x=-\mu u_x\), with \(z^t\) and \(z^x\) constant.

Thus the linear loss term is incorporated directly into the reaction part of the Hamiltonian, without requiring any explicit dependence on the contact variables.

\subsection{A damped quartic field on a Minkowski space}
\label{subsec:phi4_31}

Let \(x^0=t\) and \(x^1,x^2,x^3\) be the standard coordinates on \(\mathbb R^{3,1}\). Consider the damped \(\phi^4\) equation
\begin{equation}\label{eq:damped_phi4}
\phi_{tt}-\Delta\phi+m^2\phi+g\phi^3+\lambda \phi_t=0,
\qquad m>0,\ g>0,\ \lambda\ge0,
\end{equation}
which may be regarded as a dissipative nonlinear Klein--Gordon model with quartic potential \cite{Ryd_96,WangCheng2005,QuinteroSanchezMertens2001,LizunovaVanWezel2021}.

We describe \eqref{eq:damped_phi4} on the canonical four-contact manifold \(\mathcal J_{\mathbb R,4}=\bigoplus^4 \T^*\mathbb R\times\mathbb R^4\) with coordinates \((\phi,p^0,p^1,p^2,p^3,z^0,z^1,z^2,z^3)\) and canonical four-contact form \(\bm\eta_{\mathbb{R},4}=\sum_{\alpha=0}^3(\dd z^\alpha-p^\alpha \dd\phi)\otimes e_\alpha.\) Define the Hamiltonian
\[
h_{\phi^4}
=
\frac12\Bigl((p^0)^2-(p^1)^2-(p^2)^2-(p^3)^2\Bigr)
+\frac{m^2}{2}\phi^2+\frac{g}{4}\phi^4+\lambda z^0.
\]
The associated Hamilton--De Donder--Weyl equations are
\[
\partial_0\phi=\frac{\partial h_{\phi^4}}{\partial p^0}=p^0,\qquad
\partial_i\phi=\frac{\partial h_{\phi^4}}{\partial p^i}=-p^i,\qquad i=1,2,3,
\]
\[
\sum_{\mu=0}^3\partial_\mu p^\mu
=
-\left(
\frac{\partial h_{\phi^4}}{\partial \phi}
+\sum_{\mu=0}^3 p^\mu \frac{\partial h_{\phi^4}}{\partial z^\mu}
\right)
=
-\bigl(m^2\phi+g\phi^3+\lambda p^0\bigr),
\]
and
\[
\sum_{\mu=0}^3\partial_\mu z^\mu
=
\sum_{\mu=0}^3 p^\mu \frac{\partial h_{\phi^4}}{\partial p^\mu}-h_{\phi^4}
=
\frac12\Bigl((p^0)^2-(p^1)^2-(p^2)^2-(p^3)^2\Bigr)
-\frac{m^2}{2}\phi^2-\frac{g}{4}\phi^4-\lambda z^0.
\]
The regularity of $h_{\phi^4}$ allows us to describe the momenta in terms of $\phi_t,\phi_x,\phi$ and relate $h_{\phi^4}$ to a second-order system of PDEs on the variable $\phi$. Indeed, using \(p^0=\phi_t\) and \(p^i=-\partial_i\phi\), \(i=1,2,3\), the momentum balance equation yields precisely \eqref{eq:damped_phi4}. Hence, \(h_{\phi^4}\) provides a natural \(4\)-contact Hamilton--De Donder--Weyl description of the damped \(\phi^4\) field.

Moreover, the dissipative variables satisfy
\[
\sum_{\mu=0}^3\partial_\mu z^\mu
=
\frac12\bigl(\phi_t^2-|\nabla\phi|^2\bigr)
-\frac{m^2}{2}\phi^2-\frac{g}{4}\phi^4-\lambda z^0.
\]
Thus, the dissipation is encoded through the time-like contact variable \(z^0\), while the remaining \(z^i\) contribute only through the divergence term.

\subsection{Damped sine--Gordon equation}
\label{subsec:damped_sG}

Consider the damped sine--Gordon equation
\begin{equation}\label{eq:damped_sine_gordon}
u_{tt}-c^2u_{xx}+\sin u+\lambda u_t=0,
\qquad c>0,\ \lambda\ge0,
\end{equation}
which arises naturally in dissipative models of long Josephson junctions and, more broadly, in applied versions of the sine--Gordon model \cite{McLaughlinScott1978,CostabileEtAl1978,CuevasMaraverKevrekidisWilliams2014}.

We describe \eqref{eq:damped_sine_gordon} on the canonical two-contact manifold \(\mathcal J_{\mathbb R,2}=\bigoplus^2\T^*\mathbb  R\times\mathbb R^2\) with coordinates \((u,p^t,p^x,z^t,z^x)\) and the standard two-contact form $\bm\eta_{\mathbb{R},2}$. Consider the Hamiltonian
\[
h_{\mathrm{sG}}=\frac12\Bigl((p^t)^2-\frac{1}{c^2}(p^x)^2\Bigr)+(1-\cos u)+\lambda z^t.
\]
The Hamilton--De Donder--Weyl equations read
\[
u_t=\frac{\partial h_{\mathrm{sG}}}{\partial p^t}=p^t,\qquad
u_x=\frac{\partial h_{\mathrm{sG}}}{\partial p^x}=-\frac{1}{c^2}p^x,
\]
and
\[
\partial_t p^t+\partial_x p^x
=
-\left(
\frac{\partial h_{\mathrm{sG}}}{\partial u}
+p^t\frac{\partial h_{\mathrm{sG}}}{\partial z^t}
+p^x\frac{\partial h_{\mathrm{sG}}}{\partial z^x}
\right)
=
-\bigl(\sin u+\lambda p^t\bigr).
\]
Moreover, the balance equation for the dissipative variables is
\[
\partial_t z^t+\partial_x z^x
=
p^t\frac{\partial h_{\mathrm{sG}}}{\partial p^t}
+p^x\frac{\partial h_{\mathrm{sG}}}{\partial p^x}
-h_{\mathrm{sG}}
=
\frac12\Bigl((p^t)^2-\frac{1}{c^2}(p^x)^2\Bigr)-(1-\cos u)-\lambda z^t.
\]
Since $h_{\rm sG}$ is regular, one can describe the momenta by means of the derivatives $u_t,u_x$ and remaining coordinates.  Using \(p^t=u_t\) and \(p^x=-c^2u_x\), the momentum balance equation yields precisely \eqref{eq:damped_sine_gordon}. Hence, \(h_{\mathrm{sG}}\) provides a two-contact Hamilton--De Donder--Weyl formulation of the damped sine--Gordon equation. Furthermore, the dissipative variables satisfy
\[
\partial_t z^t+\partial_x z^x
=
\frac12\bigl(u_t^2-c^2u_x^2\bigr)-(1-\cos u)-\lambda z^t.
\]

One could use a more general approach to include a damped sine-Gordon equation with dissipation modelled by a function $\lambda(t,x)$. Indeed, let us define Hamiltonian
\[
h_{\mathrm{sG}}^{\mathrm{quad}}
=
\frac12\Bigl((p^t)^2-\frac{1}{c^2}(p^x)^2\Bigr)+(1-\cos u)+\frac12 (z^t)^2.
\]
Then, the Hamilton--De Donder--Weyl equations become
\[
u_t=p^t,\qquad u_x=-\frac{1}{c^2}p^x,
\qquad 
\partial_t p^t+\partial_x p^x=-(\sin u+z^t p^t),
\]
and
\[
\partial_t z^t+\partial_x z^x
=
\frac12\Bigl((p^t)^2-\frac{1}{c^2}(p^x)^2\Bigr)-(1-\cos u)-\frac12(z^t)^2.
\]
Hence, after eliminating the momenta, which is possible due to the regularity of $h^{\rm quad}_{\rm sG}$, one obtains
\[
u_{tt}-c^2u_{xx}+\sin u+z^t u_t=0.
\]
Moreover, the form of $h^{\rm quad}_{\rm sG}$ shows that the equation is hyperbolic. 

This formulation is useful if one wishes to realise a prescribed damping profile \(\lambda=\lambda(x,t)\). Indeed, imposing
\(
z^t=\lambda(x,t),
\)
the field equation becomes
\[
u_{tt}-c^2u_{xx}+\sin u+\lambda(x,t)u_t=0,
\]
while the balance equation for the dissipative variables reduces to
\[
\partial_x z^x
=
\frac12\bigl(u_t^2-c^2u_x^2\bigr)-(1-\cos u)-\frac12\lambda(x,t)^2-\partial_t\lambda(x,t).
\]
Therefore, once \(u\) and \(\lambda(x,t)\) are fixed, the variable \(z^x\) can be chosen locally so that the above identity holds. In this sense, the quadratic dependence on \(z^t\) allows one to accommodate nonconstant damping profiles by adjusting the second dissipative variable \(z^x\).
 
\subsection{FitzHugh--Nagumo reaction--diffusion system}
\label{subsec:FHN}

Consider the FitzHugh--Nagumo system for two dependent variables $u=u(t,x)$ and $v=v(t,x)$ of the form
\begin{equation}\label{eq:FHN}
u_t=D_u u_{xx}+f(u)-v+I,\qquad
v_t=D_v v_{xx}+\varepsilon(u-a v),
\end{equation}
where \(D_u,D_v>0\), \(\varepsilon>0\), \(a>0\), $f(u)$ is a certain function, and \(I\in\mathbb R\). This system was originally introduced as a simplification of the Hodgkin--Huxley model and is a standard prototype for excitable media \cite{FitzHugh1961,Hastings1982FHN}.

Let us describe \eqref{eq:FHN} on a two-contactification. Consider \(P_0=\mathbb R^8\) with global coordinates \((u,v,r,s,p_u^x,p_v^x,q_u^x,q_v^x)\), and define
\[
\bm \theta =\frac12(-r\,\dd u+u\,\dd r-s\,\dd v+v\,\dd s)\otimes e_1+(p_u^x\,\dd u+p_v^x\,\dd v+q_u^x\,\dd r+q_v^x\,\dd s)\otimes e_2.
\]
This gives rise to an exact  two-symplectic manifold $(P_0,\bm\omega=-\dd\bm\theta)$ with \[\bm \omega=(-\dd u\wedge \dd r-\dd v\wedge \dd s)\otimes e_1+(\dd u\wedge \dd p_u^x+\dd v\wedge \dd p_v^x+\dd r\wedge \dd q_u^x+\dd s\wedge \dd q_v^x)\otimes e_2.\] Consider the contactification \(M=P_0\times\mathbb R^2\) with coordinates \((u,v,r,s,p_u^x,p_v^x,q_u^x,q_v^x,z^t,z^x)\) and two-contact form 
$$
{\bm \eta}=(\dd z^t-\rho^*\theta_e)\otimes e_1+(\dd z^x-\rho^*\theta_{P_0})\otimes e_2
$$
with  the projection $\rho:M\rightarrow P_0$. Define the Hamiltonian $h_{\mathrm{FHN}}\in C^\infty(P_0\times \mathbb{R}^2)$ of the form \[ h_{\mathrm{FHN}}= -\frac1{D_u}p_u^x q_u^x-\frac1{D_v}p_v^x q_v^x-r\,(f(u)-v+I)-s\,\varepsilon(u-a v). \]

Since \(h_{\mathrm{FHN}}\) is independent of \(z^t\) and \(z^x\), the Hamilton--De Donder--Weyl equations reduce to
\[
\partial_t r-\partial_x p_u^x=\frac{\partial h_{\mathrm{FHN}}}{\partial u},\qquad
-\partial_t u-\partial_x q_u^x=\frac{\partial h_{\mathrm{FHN}}}{\partial r},
\]
\[
\partial_t s-\partial_x p_v^x=\frac{\partial h_{\mathrm{FHN}}}{\partial v},\qquad
-\partial_t v-\partial_x q_v^x=\frac{\partial h_{\mathrm{FHN}}}{\partial s},
\]
together with \(u_x=\partial h_{\mathrm{FHN}}/\partial p_u^x=-q_u^x/D_u\) and \(v_x=\partial h_{\mathrm{FHN}}/\partial p_v^x=-q_v^x/D_v\). Hence, \(q_u^x=-D_u u_x\) and \(q_v^x=-D_v v_x\). Moreover,
\[
\frac{\partial h_{\mathrm{FHN}}}{\partial r}=-(f(u)-v+I),\qquad
\frac{\partial h_{\mathrm{FHN}}}{\partial s}=-\varepsilon(u-a v),
\]
and the equations for \(u\) and \(v\) become
\[
u_t=-\partial_x q_u^x+f(u)-v+I,\qquad
v_t=-\partial_x q_v^x+\varepsilon(u-a v).
\]
Substituting the expressions for \(q_u^x\) and \(q_v^x\), one recovers precisely \eqref{eq:FHN}. Therefore, the FitzHugh--Nagumo system admits a natural two-contact Hamilton--De Donder--Weyl description on a reduced phase space.

The remaining equations determine the auxiliary fields \(r\) and \(s\). In particular, \(r=s=p_u^x=p_v^x=0\) is a consistent reduction, and every solution \((u,v)\) of \eqref{eq:FHN} lifts to a solution of the full two-contact Hamiltonian system by taking \(q_u^x=-D_u u_x\), \(q_v^x=-D_v v_x\), with \(z^t\) and \(z^x\) constant.

Let us finally describe a simple infinitesimal two-contact symmetry in a particular affine FitzHugh--Nagumo-type case. Assume that \(f(u)=a^{-1}u+c\), with \(c\in\mathbb R\). Then
\[
h=-\frac1{D_u}p_u^xq_u^x-\frac1{D_v}p_v^xq_v^x
-r\left(\frac1a u+c-v+I\right)-s\varepsilon(u-av).
\]
Consider the vector field
\[
\Gamma=
\frac{\partial}{\partial u}
+\frac1a\frac{\partial}{\partial v}
+\frac12\left(r+\frac{s}{a}\right)\frac{\partial}{\partial z^t}.
\]
Since
\[
L_{\partial_u+a^{-1}\partial_v}\theta^t
=d\left[\frac12\left(r+\frac{s}{a}\right)\right],
\qquad
L_{\partial_u+a^{-1}\partial_v}\theta^x=0,
\]
one has \(\mathcal{L}_\Gamma\eta=0\). Moreover, a short calculation shows that \(\Gamma(h)=0\). The corresponding current \(F=-\iota_\Gamma\bm\eta\) is
\[
F=
-\left(r+\frac{s}{a}\right)\otimes e_1
+\left(p_u^x+\frac1a p_v^x\right)\otimes e_2.
\]
Since \(R_t h=R_xh=0\), this gives the conservation law
\[
\partial_t\left[-\left(r+\frac{s}{a}\right)\right]
+\partial_x\left(p_u^x+\frac1a p_v^x\right)=0
\]
along every solution of the associated two-contact Hamilton--De Donder--Weyl equations.

\section{Conclusions}\label{sec:conclusions}

The results of this work show that \(k\)-contact geometry provides an effective and conceptually coherent Hamiltonian framework for a substantial class of dissipative field equations. Beyond the original motivating examples, the paper establishes two complementary geometric mechanisms for constructing nonconservative PDEs from \(k\)-contact Hamiltonian systems: canonical realizations on \(\bigoplus^k \T^*Q\times\mathbb R^k\), which naturally capture damped models and Hamiltonians with explicit dependence on the dissipative variables, and reduced realizations on $k$-contactifications of exact \(k\)-symplectic phase spaces, which are particularly well adapted to nonlinear diffusion, reaction--diffusion systems, complex amplitude equations, and excitable-media models. In this way, the paper does not merely collect isolated examples, but identifies a reproducible geometric strategy that explains why equations of very different analytic nature can be treated within the same \(k\)-contact Hamilton--De Donder--Weyl formalism.

The additional structural results obtained in Section~\ref{sec:tools_apps}, including the splitting of the HdDW equations on contactifications, the notion of adapted reduced two-contact phase space, the corresponding partial regularity condition, and the formulation of dissipation laws associated with infinitesimal dynamical symmetries, clarify the internal geometry behind these constructions and show how second-order dissipative PDEs emerge from first-order \(k\)-contact systems after elimination of the appropriate momenta. In particular, the present version of the paper goes beyond a purely descriptive collection of examples: it develops a general mechanism for constructing reduced two-contact models and shows how dissipation laws can be recovered geometrically from dynamical symmetries.

The worked examples of Section~\ref{sec:examples}, summarised in Table~\ref{tab:kcontact_examples}, show that this formalism already applies to a significant family of physically relevant equations, ranging from damped Klein--Gordon, sine--Gordon, and \(\phi^4\) models to Allen--Cahn, generalized Burgers, porous medium, Fisher--KPP, complex Ginzburg--Landau, damped nonlinear Schr\"odinger, and FitzHugh--Nagumo systems. Moreover, the paper also exhibits concrete instances in which the formalism geometrizes balance laws already familiar in the PDE literature, such as momentum-type dissipation laws for damped wave equations, and produces nontrivial conserved or dissipated quantities associated with symmetries in reduced models. This strengthens the interpretation of \(k\)-contact geometry not only as a language for rewriting dissipative PDEs, but as a framework capable of encoding structurally meaningful information about their dynamics.

Our formalism already covers a substantial class of dissipative models, while the candidate equations listed in Table \ref{tab:kcontact_applications_II} strongly suggest that its range of applicability is much wider. For this reason, the paper should be viewed both as a collection of explicit applications and as a foundational step towards further developments, symmetry reduction, Hamilton--Jacobi methods, and the geometric analysis of more singular or irreversible continuum models.

The models in Table  \ref{tab:kcontact_applications_II} are not presented here as fully established $k$-contact theories, but rather as natural candidate systems whose structure suggests that they may admit a meaningful description within the $k$-contact framework.

\begin{table}[t]
\centering
\small
\renewcommand{\arraystretch}{1.15}
\setlength{\tabcolsep}{4pt}
\begin{tabularx}{\textwidth}{@{}L{2.9cm}L{2.6cm}L{5.7cm}X@{}}
\toprule
Model & Configuration space $Q$ & Representative Hamiltonian $h$ & Target equation \\
\midrule
Damped sine--Gordon field
&
$\mathbb{R}$, coordinate $u$
&
$h=\frac12\bigl((p_u^t)^2-c^{-2}(p_u^x)^2\bigr)+(1-\cos u)+\lambda z^t$
&
$u_{tt}-c^2u_{xx}+\sin u+\lambda u_t=0$.
\\

Damped double sine--Gordon field
&
$\mathbb{R}$, coordinate $u$
&
$h=\frac12\bigl((p_u^t)^2-c^{-2}(p_u^x)^2\bigr)+a(1-\cos u)+b(1-\cos 2u)+\lambda z^t$
&
$u_{tt}-c^2u_{xx}+a\sin u+2b\sin(2u)+\lambda u_t=0$.
\\

Damped $\phi^4$ field
&
$\mathbb{R}$, coordinate $u$
&
$h=\frac12\bigl((p_u^t)^2-c^{-2}(p_u^x)^2\bigr)+\frac{m^2}{2}u^2+\frac{\beta}{4}u^4+\lambda z^t$
&
$u_{tt}-c^2u_{xx}+m^2u+\beta u^3+\lambda u_t=0$.
\\

Coupled dissipative scalar fields
&
$\mathbb{R}^2$, coordinates $(u,v)$
&
$h=\frac12\bigl((p_u^t)^2-c_u^{-2}(p_u^x)^2+(p_v^t)^2-c_v^{-2}(p_v^x)^2\bigr)+V(u,v)+\lambda z^t$
&
$u_{tt}-c_u^2u_{xx}+\partial_uV+\lambda u_t=0$, \ 
$v_{tt}-c_v^2v_{xx}+\partial_vV+\lambda v_t=0$.
\\

Driven Josephson-type junction
&
$\mathbb{R}$, coordinate $u$
&
$h=\frac12\bigl((p_u^t)^2-c^{-2}(p_u^x)^2\bigr)+(1-\cos u)-Iu+\lambda z^t$
&
$u_{tt}-c^2u_{xx}+\sin u+\lambda u_t=I(t,x)$.
\\
Dissipative nonlinear sigma model
&
A Riemannian manifold $(N,g)$
&
$h=\frac12 g^{AB}(u)\bigl(p_A^tp_B^t-c^{-2}p_A^xp_B^x\bigr)+U(u)+\lambda z^t$
&
Schematic damped wave-map equation with potential:
$\Box_g u+\nabla U(u)+\lambda u_t=0$.
\\

Damped massive Klein--Gordon system
&
$\mathbb{R}^N$, coordinates $(u^1,\dots,u^N)$
&
$h=\frac12\sum_{A=1}^N\bigl((p_A^t)^2-c_A^{-2}(p_A^x)^2+m_A^2(u^A)^2\bigr)+\lambda z^t$
&
$u^A_{tt}-c_A^2u^A_{xx}+m_A^2u^A+\lambda u^A_t=0$,
\quad $A=1,\dots,N$.
\\

Coupled damped Klein--Gordon system
&
$\mathbb{R}^N$, coordinates $(u^1,\dots,u^N)$
&
$h=\frac12\sum_{A=1}^N\bigl((p_A^t)^2-c_A^{-2}(p_A^x)^2\bigr)+V(u^1,\dots,u^N)+\lambda z^t$
&
$u^A_{tt}-c_A^2u^A_{xx}+\partial_{u^A}V+\lambda u^A_t=0$,
\quad $A=1,\dots,N$.
\\

Damped anisotropic wave equation
&
$\mathbb{R}$, coordinate $u$
&
$h=\frac12\Bigl((p_u^t)^2-\sum_{a=1}^d c_a^{-2} (p_u^{x^a})^2\Bigr)+V(u)+\lambda z^t$
&
$u_{tt}-\sum_{a=1}^d c_a^2u_{x^ax^a}+V'(u)+\lambda u_t=0$.
\\

Damped nonlinear telegraph equation with variable potential
&
$\mathbb{R}$, coordinate $u$
&
$h=\frac12\bigl((p_u^t)^2-c^{-2}(p_u^x)^2\bigr)+W(u)+\lambda z^t$
&
$u_{tt}-c^2u_{xx}+W'(u)+\lambda u_t=0$,
including damped telegrapher and nonlinear wave models.
\\
\bottomrule
\end{tabularx}
\caption{Possible additional applications of the $k$-contact formalism.}
\label{tab:kcontact_applications_II}
\end{table}


\FloatBarrier
\printbibliography

@article{AllenCahn1979,
  author  = {Allen, S. M. and Cahn, J. W.},
  title   = {A Microscopic Theory for Antiphase Boundary Motion and Its Application to Antiphase Domain Coarsening},
  journal = {Acta Metallurgica},
  shortjournal = {Acta Metall.},
  volume  = {27},
  number  = {6},
  pages   = {1085--1095},
  year    = {1979},
  doi     = {10.1016/0001-6160(79)90196-2}
}

@article{Hastings1982FHN,
  author  = {Hastings, Stuart P.},
  title   = {Single and Multiple Pulse Waves for the FitzHugh--Nagumo Equations},
  journal = {SIAM Journal on Applied Mathematics},
  shortjournal = {SIAM J. Appl. Math.},
  volume  = {42},
  number  = {2},
  pages   = {247--260},
  year    = {1982},
  month   = apr,
  publisher = {Society for Industrial and Applied Mathematics},
  doi     = {10.1137/0142019}
}

@article{KasTanHataTaka2001,
  author  = {Kasai, Y. and Tanda, S.  and Hatakenaka, N. and Takayanagi, H.},
  title   = {Fluxon dynamics in isolated long Josephson junctions},
  journal = {Physica C: Superconductivity},
  shortjournal = {Physica C},
  volume  = {352},
  number  = {1--4},
  pages   = {211--214},
  year    = {2001},
  doi     = {10.1016/S0921-4534(00)01727-5}
}

@article{Popov2005,
  author  = {Popov, Constantine A.},
  title   = {Perturbation theory for the double-sine-Gordon equation},
  journal = {Wave Motion},
  shortjournal = {Wave Motion},
  volume  = {42},
  number  = {4},
  pages   = {309--316},
  year    = {2005},
  doi     = {10.1016/j.wavemoti.2005.04.007}
}

@article{KayaElSayed2004,
  author  = {Kaya, D. and El-Sayed, S. M.},
  title   = {A numerical simulation and explicit solutions of the generalized Burgers--Fisher equation},
  journal = {Applied Mathematics and Computation},
  shortjournal = {Appl. Math. Comput.},
  volume  = {152},
  number  = {2},
  pages   = {403--413},
  year    = {2004},
  doi     = {10.1016/S0096-3003(03)00565-4}
}

@article{PopovicPtashnykSattar2025,
  author  = {Popovi{\'c}, Nikola and Ptashnyk, Mariya and Sattar, Zak},
  title   = {The Burgers-FKPP advection-reaction-diffusion equation with cut-off},
  journal = {Journal of Dynamics and Differential Equations},
  shortjournal = {J. Dynam. Differential Equations},
  year    = {2025},
  doi     = {10.1007/s10884-025-10458-y}
}

@misc{Salih2016BurgersEquation,
  author       = {Salih, A.},
  title        = {Burgers' Equation},
  institution  = {Indian Institute of Space Science and Technology},
  location     = {Thiruvananthapuram},
  year         = {2016},
  month        = feb,
  day          = {18},
  note         = {Lecture notes, Department of Aerospace Engineering},
  url          = {https://iist.cygnusdvlp.in/sites/default/files/2025-06/Burgers_equation_viscous.pdf}
}

@article{OP2020,
    author = {Opanasenko, Stanislav and Popovych, Roman O.},
    title = {Generalized symmetries and conservation laws of (1 + 1)-dimensional Klein–Gordon equation},
    journal = {Journal of Mathematical Physics},
    shortjournal = {J. Math. Phys.},
    volume = {61},
    number = {10},
    pages = {101515},
    year = {2020},
    month = {10},
    issn = {0022-2488},
    doi = {10.1063/5.0003304},
    url = {https://doi.org/10.1063/5.0003304},
    eprint = {https://pubs.aip.org/aip/jmp/article-pdf/doi/10.1063/5.0003304/16167317/101515_1_online.pdf},
}

@misc{deLucasLangeSardonRivas2026,
  author       = {de Lucas, Javier and Lange, Jan and Sardon, Carlos and Rivas, Xavier},
  title        = {Hamilton–Jacobi theory for non-conservative field theories in the $k$-contact framework},
  year         = {2026},
  archiveprefix= {arXiv},
  eprint       = {2604.27670},
  note         = {Preprint}
}

@article{TakacsWagner2006,
  author  = { Tak{\'a}cs, G. and W\'agner, F.},
  title   = {Double sine-Gordon model revisited},
  journal = {Nuclear Physics B},
  shortjournal = {Nuclear Phys. B},
  volume  = {741},
  number  = {3},
  pages   = {353--367},
  year    = {2006},
  doi     = {10.1016/j.nuclphysb.2006.02.004}
}

@article{FitzHugh1961,
  author  = {FitzHugh, Richard},
  title   = {Impulses and Physiological States in Theoretical Models of Nerve Membrane},
  journal = {Biophysical Journal},
  shortjournal = {Biophys. J.},
  volume  = {1},
  number  = {6},
  pages   = {445--466},
  year    = {1961},
  doi     = {10.1016/S0006-3495(61)86902-6}
}

@article{Nagumo1962,
  author  = {Nagumo, Jin-Ichi and Arimoto, Suguru and Yoshizawa, Shuji},
  title   = {An Active Pulse Transmission Line Simulating Nerve Axon},
  journal = {Proceedings of the IRE},
  shortjournal = {Proc. IRE},
  volume  = {50},
  number  = {10},
  pages   = {2061--2070},
  year    = {1962},
  doi     = {10.1109/JRPROC.1962.288235}
}

@article{AransonKramer2002,
  author  = {Aranson, Igor S. and Kramer, Lorenz},
  title   = {The World of the Complex Ginzburg--Landau Equation},
  journal = {Reviews of Modern Physics},
  shortjournal = {Rev. Modern Phys.},
  volume  = {74},
  number  = {1},
  pages   = {99--143},
  year    = {2002},
  doi     = {10.1103/RevModPhys.74.99}
}

@article{McLaughlinScott1978,
  author  = {McLaughlin, D. W. and Scott, A. C.},
  title   = {Perturbation Analysis of Fluxon Dynamics},
  journal = {Physical Review A},
  shortjournal = {Phys. Rev. A},
  volume  = {18},
  number  = {4},
  pages   = {1652--1680},
  year    = {1978},
  doi     = {10.1103/PhysRevA.18.1652}
}

@article{CostabileEtAl1978,
  author  = {Costabile, G. and Parmentier, R. D. and Savo, B. and McLaughlin, D. W. and Scott, A. C.},
  title   = {Exact Solutions of the Sine-Gordon Equation Describing Oscillations in a Long (but Finite) Josephson Junction},
  journal = {Applied Physics Letters},
  shortjournal = {Appl. Phys. Lett.},
  volume  = {32},
  number  = {9},
  pages   = {587--589},
  year    = {1978},
  doi     = {10.1063/1.90113}
}

@book{CuevasMaraverKevrekidisWilliams2014,
  editor    = {Cuevas-Maraver, Jes{\'u}s and Kevrekidis, Panayotis G. and Williams, Floyd},
  title     = {The sine-Gordon Model and its Applications: From Pendula and Josephson Junctions to Gravity and High-Energy Physics},
  publisher = {Springer},
  location  = {Cham},
  year      = {2014},
  doi       = {10.1007/978-3-319-06722-3}
}

@article{WangCheng2005,
  author  = {Wang, QuanFang and Cheng, DaiZhan},
  title   = {Numerical Solution of Damped Nonlinear Klein--Gordon Equations Using Variational Method and Finite Element Approach},
  journal = {Applied Mathematics and Computation},
  shortjournal = {Appl. Math. Comput.},
  volume  = {162},
  number  = {1},
  pages   = {381--401},
  year    = {2005},
  doi     = {10.1016/j.amc.2003.12.102}
}

@book{Murray2002,
  author    = {Murray, James D.},
  title     = {Mathematical Biology. I: An Introduction},
  edition   = {3},
  publisher = {Springer},
  location  = {New York},
  year      = {2002},
  doi       = {10.1007/b98868}
}

@article{KuramotoTsuzuki1976,
  author  = {Kuramoto, Yoshiki and Tsuzuki, Toshio},
  title   = {Persistent Propagation of Concentration Waves in Dissipative Media Far from Thermal Equilibrium},
  journal = {Progress of Theoretical Physics},
  shortjournal = {Progr. Theoret. Phys.},
  volume  = {55},
  number  = {2},
  pages   = {356--369},
  year    = {1976},
  doi     = {10.1143/PTP.55.356}
}

@article{CrossHohenberg1993,
  author  = {Cross, M. C. and Hohenberg, P. C.},
  title   = {Pattern Formation Outside of Equilibrium},
  journal = {Reviews of Modern Physics},
  shortjournal = {Rev. Modern Phys.},
  volume  = {65},
  number  = {3},
  pages   = {851--1112},
  year    = {1993},
  doi     = {10.1103/RevModPhys.65.851}
}

@article{LugiatoLefever1987, author = {Lugiato, L. A. and Lefever, R.}, title = {Spatial Dissipative Structures in Passive Optical Systems}, journal = {Physical Review Letters},
 shortjournal = {Phys. Rev. Lett.}, volume = {58}, number = {21}, pages = {2209--2211}, year = {1987}, doi = {10.1103/PhysRevLett.58.2209} }

@article{KippenbergEtAl2018, author = {Kippenberg, Tobias J. and Gaeta, Alexander L. and Lipson, Michal and Gorodetsky, Mikhail L.}, title = {Dissipative Kerr Solitons in Optical Microresonators}, journal = {Science},
 shortjournal = {Science}, volume = {361}, pages = {6402}, year = {2018}, doi = {10.1126/science.aan8083} }

@book{ErmentroutTerman2010,
  author    = {Ermentrout, G. Bard and Terman, David H.},
  title     = {Mathematical Foundations of Neuroscience},
  publisher = {Springer},
  year      = {2010},
  doi       = {10.1007/978-0-387-87708-2}
}

@book{Vazquez2007,
  author    = {V{\'a}zquez, Juan Luis},
  title     = {The Porous Medium Equation: Mathematical Theory},
  publisher = {Oxford University Press},
  year      = {2007},
  isbn      = {9780198569039}
}

@article{Barenblatt1952,
  author  = {Barenblatt, G. I.},
  title   = {On Some Unsteady Motions of a Liquid or a Gas in a Porous Medium},
  journal = {Prikladnaya Matematika i Mekhanika},
  shortjournal = {Prikl. Mat. Mekh.},
  volume  = {16},
  number  = {1},
  pages   = {67--78},
  year    = {1952}
}

@book{Fife1979,
  author    = {Fife, Paul C.},
  title     = {Mathematical Aspects of Reacting and Diffusing Systems},
  series    = {Lecture Notes in Biomathematics},
  volume    = {28},
  publisher = {Springer},
  location  = {Berlin},
  year      = {1979},
  doi       = {10.1007/978-3-540-34932-7}
}

@article{Burgers1948,
  author  = {Burgers, J. M.},
  title   = {A Mathematical Model Illustrating the Theory of Turbulence},
  journal = {Advances in Applied Mechanics},
  shortjournal = {Adv. Appl. Mech.},
  volume  = {1},
  pages   = {171--199},
  year    = {1948},
  doi     = {10.1016/S0065-2156(08)70100-5}
}

@incollection{Aronson1986,
  author    = {Aronson, D. G.},
  title     = {The Porous Medium Equation},
  booktitle = {Some Problems in Nonlinear Diffusion},
  editor    = {Fasano, A. and Primicerio, M.},
  series    = {Lecture Notes in Mathematics},
  volume    = {1224},
  pages     = {1--46},
  publisher = {Springer},
  location  = {Berlin},
  year      = {1986}
}

@book{SulemSulem1999,
  author    = {Sulem, Catherine and Sulem, Pierre-Louis},
  title     = {The Nonlinear Schr{\"o}dinger Equation: Self-Focusing and Wave Collapse},
  publisher = {Springer},
  year      = {1999},
  doi       = {10.1007/b98958}
}

@article{Fisher1937,
  author  = {Fisher, Ronald A.},
  title   = {The Wave of Advance of Advantageous Genes},
  journal = {Annals of Eugenics},
  shortjournal = {Ann. Eugenics},
  volume  = {7},
  pages   = {355--369},
  year    = {1937},
  doi     = {10.1111/j.1469-1809.1937.tb02153.x}
}

@article{Kolmogorov1937,
  author  = {Kolmogorov, A. N. and Petrovskii, I. G. and Piskunov, N. S.},
  title   = {A Study of the Diffusion Equation with Increase in the Amount of Substance, and its Application to a Biological Problem},
  journal = {Bulletin of Moscow University, Mathematics and Mechanics},
  shortjournal = {Moscow Univ. Math. Bull.},
  volume  = {1},
  pages   = {1--25},
  year    = {1937}
}

@book{Arn_89,
    author = {Vladimir I. Arnold},
    title = {{Mathematical Methods of Classical Mechanics}},
    address = {New York},
    isbn = {0387968903},
    publisher = {Springer},
    series = {Graduate Texts in Mathematics},
    volume = {60},
    year = {1989},
    note = {\href{https://doi.org/10.1007/978-1-4757-1693-1}{10.1007/978-1-4757-1693-1}}
}

@article{Bra_17,
    author = {Alessandro Bravetti},
    title = {{Contact Hamiltonian dynamics: The concept and its use}},
    journal = {Entropy},
    shortjournal = {Entropy},
    volume = {19},
    number = {10},
    pages = {535},
    year = {2017},
    doi = {10.3390/e19100535}
}

@article{BCT_17,
    author = {Alessandro Bravetti and Hans Cruz and Diego Tapias},
    title = {{Contact Hamiltonian mechanics}},
    journal = {Ann. Phys.},
    shortjournal = {Ann. Physics},
    volume = {376},
    pages = {17--39},
    year = {2017},
    doi = {10.1016/j.aop.2016.11.003}
}

@article{CCM_18,
    author = {Florio M. Ciaglia and Hans Cruz and Giuseppe Marmo},
    title = {{Contact manifolds and dissipation, classical and quantum}},
    journal = {Ann. Phys.},
    shortjournal = {Ann. Physics},
    volume = {398},
    pages = {159--179},
    year = {2018},
    doi = {10.1016/j.aop.2018.09.012}
}

@article{Gun_87,
    author = {Christian Günther},
    title = {{The polysymplectic Hamiltonian formalism in field theory and calculus of variations I: The local case}},
    journal = {J. Differ. Geom.},
    shortjournal = {J. Differential Geom.},
    volume = {25},
    number = {1},
    pages = {23--53},
    year = {1987},
    doi = {10.4310/jdg/1214440723}
}

@article{GG_22,
    author = {Katarzyna Grabowska and Janusz Grabowski},
    title = {{A geometric approach to contact Hamiltonians and contact Hamilton–Jacobi theory}},
    journal = {J. Phys. A: Math. Theor.},
    shortjournal = {J. Phys. A},
    volume = {55},
    pages = {435204},
    year = {2022},
    doi = {10.1088/1751-8121/ac9adb}
}

@article{AkhmedievSotoCrespoTown2001,
  author  = {Akhmediev, Nail and Soto-Crespo, Jose M. and Town, Graham},
  title   = {Pulsating solitons, chaotic solitons, period doubling, and pulse coexistence in mode-locked lasers: Complex Ginzburg-Landau equation approach},
  journal = {Physical Review E},
  shortjournal = {Phys. Rev. E},
  volume  = {63},
  number  = {5},
  pages   = {056602},
  year    = {2001},
  doi     = {10.1103/PhysRevE.63.056602}
}

@article{KomarovLeblondSanchez2005,
  author  = {Komarov, A. and Leblond, H. and Sanchez, F.},
  title   = {Quintic complex Ginzburg-Landau model for ring fiber lasers},
  journal = {Physical Review E},
  shortjournal = {Phys. Rev. E},
  volume  = {72},
  number  = {2},
  pages   = {025604},
  year    = {2005},
  doi     = {10.1103/PhysRevE.72.025604}
}

@article{deLucasRivasSobczak2026,
  author  = {de Lucas, Javier and Rivas, Xavier and Sobczak, Tomasz},
  title   = {$k$-contact Lie systems: theory and applications},
  journal = {Geometric Mechanics},
  shortjournal = {Geometric Mechanics},
  year    = {2026},
  pages   = {1--58},
  doi     = {10.1142/S2972458926400022}
}

@article{GGMRR_20,
    author = {Jordi Gaset and Xavier Gràcia and Miguel C. Muñoz-Lecanda and Xavier Rivas and Narciso Román-Roy},
    title = {{A contact geometry framework for field theories with dissipation}},
    journal = {Ann. Phys.},
    shortjournal = {Ann. Physics},
    volume = {414},
    pages = {168092},
    year = {2020},
    doi = {10.1016/j.aop.2020.168092}
}

@article{GGMRR_21,
    author = {Jordi Gaset and Xavier Gràcia and Miguel C. Muñoz-Lecanda and Xavier Rivas and Narciso Román-Roy},
    title = {{A $k$-contact Lagrangian formulation for nonconservative field theories}},
    journal = {Rep. Math. Phys.},
    shortjournal = {Rep. Math. Phys.},
    volume = {87},
    number = {3},
    pages = {347--368},
    year = {2021},
    doi = {10.1016/S0034-4877(21)00041-0}
}

@article{GRR_22,
    title = {{Skinner--Rusk formalism for $k$-contact systems}},
    author = {Xavier Gràcia and Xavier Rivas and Narciso Román-Roy},
    journal = {J. Geom. Phys.},
    shortjournal = {J. Geom. Phys.},
    year = {2022},
    volume = {172},
    pages = {104429},
    note = {\href{https://doi.org/10.1016/j.geomphys.2021.104429}{10.1016/j.geomphys.2021.104429}}
}

@article{HLM_26,
    title = {{A k-contact geometric approach to pseudo-gauge transformations}},
    year = {2026},
    journal = {J. Phys. A: Math. Theor.},
    shortjournal = {J. Phys. A},
    author = {Hontarenko, Mykhailo and de Lucas, Javier and Maskalaniec, Adam},
    month = {2},
    doi = {10.1088/1751-8121/ae4078},
    issn = {1751-8113}
}

@article{Kij_73,
    author = {Jerzy Kijowski},
    title = {{A finite-dimensional canonical formalism in the classical field theory}},
    journal = {Comm. Math. Phys.},
    shortjournal = {Comm. Math. Phys.},
    number = {2},
    pages = {99--128},
    volume = {30},
    year = {1973},
    doi = {10.1007/BF01645975}
}

@book{KT_79,
    title = {{A symplectic framework for field theories}},
    author = {Jerzy Kijowski and Wlodzimierz M. Tulczyjew},
    year = {1979},
    volume = {107},
    series = {Lecture Notes in Physics},
    publisher = {Springer-Verlag},
    address = {Berlin Heidelberg},
    doi = {10.1007/3-540-09538-1}
}

@book{LR_89,
    title = {{Methods of Differential Geometry in Analytical Mechanics}},
    author = {Manuel de León and P. R. Rodrigues},
    publisher = {North-Holland},
    series = {Mathematics Studies},
    volume = {158},
    address = {Amsterdam},
    year = {1989},
    note = {\href{https://doi.org/10.1016/s0304-0208(08)x7115-4}{10.1016/s0304-0208(08)x7115-4}}
}

@article{LRS_25,
    title = {{Foundations on {\it k}-contact geometry}},
    year = {2024},
    author = {de Lucas, Javier and Rivas, Xavier and Sobczak, Tomasz},
    doi = {10.48550/arXiv.2409.11001},
    journal={arXiv:2409.11001},
    shortjournal = {arXiv},
    arxivId = {2409.11001}
}

@article{LS_17,
    author = {Manuel de León and Cristina Sardón},
    title = {{Cosymplectic and contact structures for time-dependent and dissipative Hamiltonian systems}},
    journal = {J. Phys. A: Math. Theor.},
    shortjournal = {J. Phys. A},
    volume = {50},
    number = {25},
    pages = {255205},
    year = {2017},
    doi = {10.1088/1751-8121/aa711d}
}

@book{LSV_15,
    author = {Manuel de León and Modesto Salgado and Silvia Vilariño},
    title = {{Methods of Differential Geometry in Classical Field Theories: $k$-symplectic and $k$-cosymplectic approaches}},
    publisher = {World Scientific},
    year = {2015},
    doi ={10.1142/9693}
}

@phdthesis{Mer_97,
    author = {Eugenio E. Merino},
    title = {{Geometría $k$-simpléctica y $k$-cosimpléctica. Aplicaciones a las teorías clásicas de campos}},
    school = {Universidade de Santiago de Compostela},
    year = {1997}
}

@article {MRSV_15,
    author = {Juan Carlos Marrero and  Narciso Román-Roy and Modesto Salgado and Silvia Vilariño},
    title = {{Reduction of polysymplectic manifolds}},
    journal = {J. Phys. A: Math. Theor.},
    shortjournal = {J. Phys. A},
    volume = {48},
    year = {2015},
    number = {5},
    pages = {055206},
    note = {\href{http://doi.org/10.1088/1751-8113/48/5/055206}{10.1088/1751-8113/48/5/055206}}
}

@phdthesis{Riv_22,
    author = {Xavier Rivas},
    title = {{Geometrical aspects of contact mechanical systems and field theories}},
    school = {Universitat Politècnica de Catalunya (UPC)},
    year = {2021},
    doi ={10.48550/arXiv.2204.11537}
}

@article{Riv_23,
    author = {Xavier Rivas},
    title = {{Nonautonomous $k$-contact field theories}},
    journal = {J. Math. Phys.},
    shortjournal = {J. Math. Phys.},
    volume = {64},
    number = {3},
    pages = {033507},
    year = {2023},
    doi = {10.1063/5.0131110}
}

@book{Ryd_96,
    author = {Lewis H. Ryder},
    title = {{Quantum Field Theory}},
    edition = {2nd},
    publisher = {Cambridge University Press},
    year = {1996},
    doi = {10.1017/CBO9780511813900}
}

@article{LizunovaVanWezel2021,
  author  = {Lizunova, Mariya and van Wezel, Jasper},
  title   = {An introduction to kinks in $\phi^4$-theory},
  journal = {SciPost Physics Lecture Notes},
  shortjournal = {SciPost Phys. Lect. Notes},
  volume  = {23},
  pages   = {1--27},
  year    = {2021},
  doi     = {10.21468/SciPostPhysLectNotes.23}
}

@article{RoquesChekroun2007,
  author  = {Roques, Lionel and Chekroun, Micka{\"e}l D.},
  title   = {On Population Resilience to External Perturbations},
  journal = {SIAM Journal on Applied Mathematics},
  shortjournal = {SIAM J. Appl. Math.},
  volume  = {68},
  number  = {1},
  pages   = {133--153},
  year    = {2007},
  doi     = {10.1137/060676994}
}

@article{OrugantiShiShivaji2002,
  author  = {Oruganti, Shobha and Shi, Junping and Shivaji, Ratnasingham},
  title   = {Diffusive Logistic Equation with Constant Yield Harvesting, {I}: Steady States},
  journal = {Transactions of the American Mathematical Society},
  shortjournal = {Trans. Amer. Math. Soc.},
  volume  = {354},
  number  = {9},
  pages   = {3601--3619},
  year    = {2002},
  doi     = {10.1090/S0002-9947-02-03005-2}
}

@article{KurataShi2008,
  author  = {Kurata, Kazuhiro and Shi, Junping},
  title   = {Optimal Spatial Harvesting Strategy and Symmetry-Breaking},
  journal = {Applied Mathematics and Optimization},
  shortjournal = {Appl. Math. Optim.},
  volume  = {58},
  number  = {1},
  pages   = {89--110},
  year    = {2008},
  doi     = {10.1007/s00245-007-9032-7}
}

@article{BandleNanbuStakgold1998,
  author  = {Bandle, Catherine and Nanbu, Tokumori and Stakgold, Ivar},
  title   = {Porous Medium Equation with Absorption},
  journal = {SIAM Journal on Mathematical Analysis},
  shortjournal = {SIAM J. Math. Anal.},
  volume  = {29},
  number  = {5},
  pages   = {1268--1278},
  year    = {1998},
  doi     = {10.1137/S0036141096311423}
}

@article{AncoMarquezGarridoGandarias2024,
  author  = {Anco, Stephen C. and M{\'a}rquez, Almudena P. and Garrido, Tamara M. and Gandarias, Maria L.},
  title   = {Conservation laws and variational structure of damped nonlinear wave equations},
  journal = {Mathematical Methods in the Applied Sciences},
  shortjournal = {Math. Methods Appl. Sci.},
  volume  = {47},
  number  = {6},
  pages   = {3974--3996},
  year    = {2024},
  doi     = {10.1002/mma.9798}
}

@article{CaiZhangWang2017,
  author  = {Cai, Wenjun and Zhang, Huai and Wang, Yushun},
  title   = {Modelling damped acoustic waves by a dissipation-preserving conformal symplectic method},
  journal = {Proceedings of the Royal Society A},
  shortjournal = {Proc. Royal Soc. A},
  volume  = {473},
  number  = {2199},
  pages   = {20160798},
  year    = {2017},
  doi     = {10.1098/rspa.2016.0798}
}

@article{Kwak1998,
  author  = {Kwak, Minkyu},
  title   = {A Porous Media Equation with Absorption. I. Long Time Behaviour},
  journal = {Journal of Mathematical Analysis and Applications},
  shortjournal = {J. Math. Anal. Appl.},
  volume  = {223},
  number  = {1},
  pages   = {96--110},
  year    = {1998},
  doi     = {10.1006/jmaa.1998.5961}
}

@article{Wang2016,
  author  = {Wang, Mingxin},
  title   = {A Diffusive Logistic Equation with a Free Boundary and Sign-Changing Coefficient in Time-Periodic Environment},
  journal = {Journal of Functional Analysis},
  shortjournal = {J. Funct. Anal.},
  volume  = {270},
  number  = {2},
  pages   = {483--508},
  year    = {2016},
  doi     = {10.1016/j.jfa.2015.10.014}
}

@article{QuinteroSanchezMertens2001,
  author  = {Quintero, Niurka R. and S{\'a}nchez, {\'A}ngel and Mertens, Franz G.},
  title   = {Anomalies of ac driven solitary waves with internal modes: Non-parametric resonances induced by parametric forces},
  journal = {Physical Review E},
  shortjournal = {Phys. Rev. E},
  volume  = {64},
  number  = {4},
  pages   = {046601},
  year    = {2001},
  doi     = {10.1103/PhysRevE.64.046601}
}

@article{Wakasugi2012,
  author  = {Wakasugi, Yuta},
  title   = {Small data global existence for the semilinear wave equation with space-time dependent damping},
  journal = {Journal of Mathematical Analysis and Applications},
  shortjournal = {J. Math. Anal. Appl.},
  volume  = {393},
  number  = {1},
  pages   = {66--79},
  year    = {2012},
  doi     = {10.1016/j.jmaa.2012.03.035}
}

@article{SakataWakasugi2017,
  author  = {Sakata, Shigehiro and Wakasugi, Yuta},
  title   = {Movement of time-delayed hot spots in Euclidean space},
  journal = {Mathematische Zeitschrift},
  shortjournal = {Math. Z.},
  volume  = {285},
  pages   = {1007--1040},
  year    = {2017},
  doi     = {10.1007/s00209-016-1735-5}
}

@article{RSS_24,
    title = {Some contributions to $k$-contact {Lagrangian} field equations, symmetries and dissipation laws},
    volume = {36},
    doi = {10.1142/S0129055X24500193},
    number = {8},
    journal = {Rev. Math. Phys.},
    shortjournal = {Rev. Math. Phys.},
    author = {Rivas, Xavier and Salgado, Modesto and Souto, Silvia},
    year = {2024},
    pages = {2450019--2450019},
}

@book{OLV_86,
    address = {New York},
    title = {Applications of {Lie} {Groups} to {Differential} {Equations}},
    volume = {107},
    isbn = {978-1-4684-0276-6},
    doi = {10.1007/978-1-4684-0274-2},
    publisher = {Springer},
    author = {Olver, Peter J.},
    year = {1986},
}

@InProceedings{SF_25,
author="Sobczak, Tomasz
and Frelik, Tymon",
editor="Nielsen, Frank
and Barbaresco, Fr{\'e}d{\'e}ric",
title="Novel Pathways in k-Contact Geometry",
booktitle="Geometric Science of Information",
year="2026",
publisher="Springer Nature Switzerland",
address="Cham",
pages="337--345",
isbn="978-3-032-03924-8"
}
\end{document}